%% file: PaperarXiv.tex
\documentclass[a4,10pt,twocolumn]{scrartcl}

\pdfoutput=1

\usepackage[top=2.5cm, bottom=2.5cm, left=1.6cm, right=1.6cm]{geometry}
\pdfoutput=1

\input{myStandardPreamble_arXiv.tex}
\usepackage[english]{babel}
\usepackage[utf8x]{inputenc}
\usepackage{enumitem}					

\usepackage{libertine}
\usepackage[T1]{fontenc}
\usepackage[noBBpl]{mathpazo}			

\usepackage[hang]{footmisc}
\setfnsymbol{wiley}

\usepackage[T1]{fontenc}        

\usepackage{graphicx}           
\usepackage{amsmath,amssymb}    
\usepackage{color}
\usepackage{xcolor}   
\usepackage{microtype} 

\usepackage{marginnote}         
\usepackage{cite}               

\usepackage{mathtools}

\usepackage{tikz}				
\usetikzlibrary{arrows}				
\usetikzlibrary{shapes.misc}
\usetikzlibrary{patterns}
\usepackage{pgfplots}
\usetikzlibrary{calc}

\renewcommand{\vec}[1]{\boldsymbol{#1}}

\makeatletter
\renewcommand{\@biblabel}[1]{[#1]\hfill}
\makeatother

\makeatletter
\let\NAT@parse\undefined
\makeatother

\usepackage[format=plain,			
            labelsep=period,		
            font=small,			
            labelfont=bf,			
            skip=5pt			
            ]{caption}
            
\usepackage{authblk}

\newcommand{\newmarkedtheorem}[1]{%
  \newenvironment{#1}
    {\pushQED{\oprocend}\csname inner@#1\endcsname}
    {\popQED\csname endinner@#1\endcsname}%
  \newtheorem{inner@#1}%
}
\newmarkedtheorem{theorem}{Theorem}
\newmarkedtheorem{lemma}{Lemma}
\newmarkedtheorem{definition}{Definition}
\newmarkedtheorem{remark}{Remark}
\newmarkedtheorem{example}{Example}
\newmarkedtheorem{assumption}{Assumption}
\newmarkedtheorem{corollary}{Corollary}
\newmarkedtheorem{proposition}{Proposition}
\newmarkedtheorem{problemFormulation}{Problem}
\newenvironment{sketchofproof}{\proof}{\endproof}

\makeatletter
\def\th@plain{%
  \thm@notefont{}
  \normalfont 
}
\def\th@definition{%
  \thm@notefont{}
  \normalfont 
}
\makeatother

\addtokomafont{paragraph}{\itshape}

\setkomafont{section}{\Large}
\setkomafont{subsection}{\large}
\setkomafont{subsubsection}{\itshape}

\input{customCommands.tex}

\title{{\LARGE \textbf
Characterizing the learning dynamics in extremum seeking} \footnote{This paper is an extended version of \cite{Wildhagen18}.} \\
\large \normalfont The role of gradient averaging and non-convexity
}

\begin{document}

\date{}
\author{Stefan Wildhagen}
\author{Simon Michalowsky}
\author{Jan Feiling}
\author[1]{Christian Ebenbauer}
\affil[1]{Institute for Systems Theory and Automatic Control, University of Stuttgart, Germany \protect\\ \texttt{\small $\lbrace$wildhagen,michalowsky,feiling,ce$\rbrace$@ist.uni-stuttgart.de}}

\maketitle

\begin{abstract}
\textbf{Abstract.} We consider perturbation-based extremum seeking, which recovers an approximate gradient of an analytically unknown objective function through measurements. Using classical needle variation analysis, we are able to explicitly quantify the recovered gradient in the scalar case. We reveal that it corresponds to an averaged gradient of the objective function, even for very general extremum seeking systems. From this, we create a recursion which represents the learning dynamics along the recovered gradient. These results give rise to the interpretation that extremum seeking actually optimizes a function other than the original one. From this insight, a new perspective on global optimization of functions with local extrema emerges: because the gradient is averaged over a certain time period, local extrema might be evened out in the learning dynamics. Moreover, a multidimensional extension of the scalar results is given.
\end{abstract}


\section{Introduction}

Extremum seeking (ES) is a well-established technique to find the extremal value of an unknown mapping \mbox{$F(\vec{x})$}. The quantity $\vec{x}$ corresponds to an input, which needs to be adjusted in such a way that it converges to a minimum of the objective function $F$. The map including an adjustment policy of $x$, such that it converges to the minimizer, is called \textit{extremum seeking system}. The only information available to this system consists of measurements of the objective function, whereas its analytical expression is unknown. As a consequence, gradient-based techniques are not suitable for this kind of problems. The idea behind ES is to recover an approximate gradient on a slower time scale by exciting the system via periodic dither signals on a faster time scale. The slow movement of the system along the recovered gradient of $F$ is referred to as \textit{learning dynamics} (LD).

Former research in this topic focused mainly on the analysis of stability of such ES schemes, whereas a quantification of this recovered gradient was given little attention. The first local stability proof of a special ES scheme was given in \cite{Krstic00}, and later extended to non-local stability by \cite{Tan06}, where it was also noted that ES approximates the gradient of $F$. A novel view on ES systems was introduced by \cite{Duerr13,Duerr17}, who also showed stability for their setup. Common to these results is that they establish so-called practical stability of the feedback scheme with respect to its parameters. This notion implies that, for convex objective functions, there exists for any set of initial conditions a choice of parameters of the ES controller, such that the ES system converges to an arbitrarily small neighborhood of a minimizer.

As mentioned above, these results are able to explain convergence to the minimizer when considering convex functions. However, global convergence in the presence of local extrema is a topic that has been rarely addressed in former research. Both \cite{Tan09} and \cite{Nesic13} presented adapted ES schemes which, under some additional (and possibly hard to verify) assumptions, achieve practical global stability with respect to their parameters, despite the local extrema.

However, it has been observed that also standard ES schemes achieve global optimization for a certain choice of parameters \cite{Tan06}. For a setup as used in \cite{Duerr13}, the LD pass through a local minimum if the frequency $\omega$ of the dither signals is chosen rather low. In \cite{Michalowsky16}, an explicit description of the recovered gradient for a particular scalar system within this setup, treated with needle-shaped dithers, was given. This result gave rise to the interpretation that the LD can be described by a gradient descent on a function other than $F$, which is parametrized by $\omega$ and will be called $L_\omega$. From this emerges the following question regarding the observation mentioned above: does $L_\omega$ become globally convex for certain parameter choices $\omega$, although $F$ is a non-convex function, such that the LD converge to the global minimizer?

The main contribution of this paper is to give an explicit description of the recovered gradient and thus the LD for general ES systems, which is valid for any parameter choice. This result contributes to the existing theory in three different ways.

First, the paper extends the results of \cite{Michalowsky16} to the use of general dithers and vector fields as considered in \cite{Grushkovskaya17}, and to multidimensional ES. We show that for all of these extensions, the LD approximately move along an \textit{averaged} gradient of $F$, as it is the case for needle-shaped inputs considered in \cite{Michalowsky16}. This is equivalent to viewing the recovered gradient as an averaged gradient of $F$.

Second, we extend the theory of ES as introduced by \cite{Duerr13}, by stating an explicit discrete-time recursion describing the LD of ES systems in this structure. Furthermore, our results explicitly state the learning descent direction with an uncertainty depending on $\omega$, whereas the existing theory did not give an interpretable estimate of this direction.

Third, we provide a new perspective on the analysis of global convergence of ES with non-convex maps: we do not aim to prove practical global stability of a new algorithm under some possibly restrictive assumptions; instead, the explicit quantification of the LD will allow us to analyze the global LD of such systems for any given parameters, and to examine when and why the LD follow the gradient of a convex function, despite $F$ being non-convex.

The paper will be organized as follows. In Section \ref{methods}, we introduce some notation and necessary background in state-transition matrices and variational calculus. Section \ref{mainresults} contains the problem statement and the theoretic main results, while a simulative verification of these statements is given in Section \ref{simulations}. Our work is summarized in Section \ref{conclusion}. This article is an extended version of \cite{Wildhagen18} and augments its contents by an introduction to variational calculus for multidimensional systems, a characterization of the recovered gradient for the multidimensional case, and an additional numerical example.

\section{Preliminaries and Methods} \label{methods}

\subsection{Notation}
We denote by $\mathbb{R}$ the set of real numbers and by $\mathbb{N}_+$ the set of positive natural numbers other than $0$. Let $\mathcal{C}^p$ denote the set of $p$ times continuously differentiable functions. We denote $\nabla f = \begin{bmatrix}
\frac{\partial f}{\partial x_1}& \dots & \frac{\partial f}{\partial x_n}\end{bmatrix}^\top$ the gradient of a function $f:\mathbb{R}^n\rightarrow\mathbb{R}, \; f\in \mathcal{C}^1$. We consider the Landau notation, i.e., for $f,g:\mathbb{R}^n\rightarrow\mathbb{R}$ write $f(\vec{x})=\mathcal{O}(g(\vec{x}))$, meaning that there exist some $M>0$ and $\delta>0$ such that $||f(\vec{x})||\le M||g(\vec{x})||$ for all $||\vec{x}||\le \delta$. Note that when \mbox{$f_1=\mathcal{O}(g)$} and $f_2=\mathcal{O}(g)$ with $M,\delta$, it also holds that \mbox{$f_1+f_2=\mathcal{O}(2g)$} with $M,\delta$, because of the triangle inequality. We denote $\vec{0}_i$ a vector in $\mathbb{R}^i$ where each entry is $0$, $\vec{I}$ the unit matrix, $a_i$ the $i$-th entry of a vector $\vec{a}$, and $B_i$ the $i$-th diagonal entry of a square matrix $\vec{B}$.

\subsection{State-Transition Matrix}

For a linear, time-varying system $\dot{\vec{x}}(t)=\vec{A}(t)\vec{x}(t)$, \mbox{$\vec{x}(t)\in\mathbb{R}^n$}, the state-transition matrix (STM) relates the solutions at different time points $\vec{x}(t)=\vec{\Phi}(t,t_0)\vec{x}(t_0)$. If $n=1$ and $A$ is locally integrable, it is defined by
\begin{equation}
\Phi(t,t_0) = \exp\left(\int_{t_0}^{t} A(\tau) \mathrm{d}\tau\right).
\end{equation}
Important properties are the so-called semi-group property $\vec{\Phi}(t,t_0)=\vec{\Phi}(t,t_1)\vec{\Phi}(t_1,t_0)$ and the differentiation property
$\frac{\mathrm{d}}{\mathrm{d} t_0}\vec{\Phi}(t,t_0) = -\vec{\Phi}(t,t_0)\vec{A}(t_0)$.

\subsection{Variational Calculus}\label{VariationalCalculus}

This section presents well-established results on the effect of so-called needle variations on the trajectories of dynamical systems, adapted to our argumentations. An exhaustive treatment of needle variations can be found in \cite{Liberzon11}.

\subsubsection{General Perturbation}

Consider the nonlinear system
\begin{equation}
\dot{\vec{x}}(t) = g_1(\vec{x}(t)) \vec{u}_1(t) + g_2(\vec{x}(t)) \vec{u}_2(t) \label{system_var}
\end{equation}
with $\vec{x}(t)\in\mathbb{R}^n$, $g_{1},g_{2}:\mathbb{R}^n \rightarrow \mathbb{R}$ and $\vec{u}_{1},\vec{u}_{2}:\mathbb{R} \rightarrow \mathbb{R}^n$. Suppose that $g_{1},g_{2} \in \mathcal{C}^1$ and $\vec{u}_{1},\vec{u}_{2}$ are piecewise continuous, such that local existence and uniqueness of the solution of \eqref{system_var} is guaranteed.

Denote $\vec{x}^*(t)$ the solution of \eqref{system_var}, when some nominal input trajectory $\vec{u}_1(t)=\vec{u}_1^*(t)$, and $\vec{u}_2(t)=\vec{u}_2^*(t)$ is applied to the system. We call $\vec{x}^*(t)$ the nominal solution of \eqref{system_var}. Next, we study the effects on the solution of \eqref{system_var}, when the nominal inputs are perturbed by a so-called needle or Pontryagin-McShane variation. We consider a perturbation in $\vec{u}_2$ only, such that the perturbed input is defined by
\begin{equation}
\vec{u}_1(t) = \vec{u}_1^*(t), \quad \vec{u}_2(t) =
\begin{cases}
\vec{u}_2^*(t) & t \notin [\bar{t},\bar{t}+\epsilon] \\
\vec{\alpha} & t \in [\bar{t},\bar{t}+\epsilon]
\end{cases}, \label{input_pert}
\end{equation}
where $\vec{\alpha}= [ \alpha_1,\ldots,\alpha_n ]^\top \in \mathbb{R}^n$ and $\bar{t},\epsilon>0$.
\begin{figure}
	\centering
	\input{Abbildungen/needleVariation.tex}
	\caption{The needle variation (left) and the nominal and perturbed trajectories (right).}
	\label{needleVar}
\end{figure}
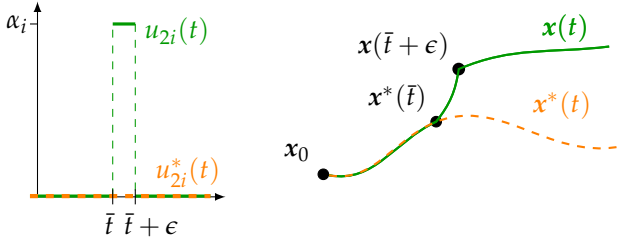
As illustrated in Fig. \ref{needleVar}, the nominal input is perturbed over an interval of length $\epsilon$, starting at $\bar{t}$, and held on some constant value $\vec{\alpha}$ in this time period. The solution of \eqref{system_var}, when the perturbed input \eqref{input_pert} is applied, will be denoted by $\vec{x}(t)$. As Fig. \ref{needleVar} suggests, $\vec{x}(t)$ will deviate from $\vec{x}^*(t)$ on $[\bar{t},\bar{t}+\epsilon]$, and then ``run parallel'' to the nominal solution in the following. By several Taylor expansions (see \cite{Liberzon11} for details), one obtains that the perturbed trajectory evolves according to
\begin{equation}
\vec{x}(t) = \vec{x}^*(t) + \epsilon \vec{v}(t) + \mathcal{O}(\epsilon^2), \quad \forall t \ge \bar{t}+\epsilon. \label{perturbed_sol}
\end{equation}
The quantity $\vec{v}(t) \in \mathbb{R}^n$ is the so-called variational variable, which evolves according to the variational equation
\begin{equation}
\dot{\vec{v}}(t) = \underbrace{
	\vec{u}_1^*(t) \nabla g_1(\vec{x}^*(t))^\top}_{\eqqcolon \vec{A}(\vec{x}^*(t),t)} \label{vareq_var}
\vec{v}(t), \quad \forall t \ge \bar{t}+\epsilon,
\end{equation}
and when $\vec{u}_2^*(\bar{t}+\epsilon)=\vec{0}$, has initial condition
\begin{equation}
\vec{v}(\bar{t}+\epsilon) = g_{2}(\vec{x}^*(\bar{t}+\epsilon))\vec{\alpha}. \label{var_IC}
\end{equation}
Note that the variational equation \eqref{vareq_var} is equivalent to a linearization of \eqref{system_var} around $\vec{x}^*(t)$, $\vec{u}_1^*(t)$, and $\vec{u}_2^*(t)$.

\subsubsection{Perturbation and Nominal Input in Single Dimension} \label{SpecPertInp}

Consider the perturbed input \eqref{input_pert} with
\begin{equation}
\vec{\alpha} = [ \vec{0}_{\ell-1} , \alpha_{\ell} , \vec{0}_{n-\ell} ]^\top, \quad \ell \in \{1,\ldots,n\},
\end{equation}
i.e., the perturbation acts only in the single dimension $\ell$. Then, the initial condition of the variational variable reads from \eqref{var_IC}
\begin{equation}
\vec{v}(\bar{t}+\epsilon) = [\vec{0}_{\ell-1} , g_{2}(\vec{x}^*(\bar{t}+\epsilon))\alpha_{\ell} , \vec{0}_{n-\ell}]^\top. \label{IC_sing}
\end{equation}
Additionally, let the nominal input $u_{1i}^*(t) = 0$ for all $i \neq \ell$ in some time period $t \in [\bar{t}+\epsilon,t_f]$. 
From the variational equation \eqref{vareq_var}, it follows that $\dot{v}_i(t)= 0$ for all $i\neq \ell$ for $t \in [\bar{t}+\epsilon,t_f]$, such that these components do not show any dynamic behavior. Furthermore, $v_i(\bar{t}+\epsilon)=0$ for all $i\neq \ell$ from \eqref{IC_sing}, such that it holds
\begin{equation}
v_i(t)=0 \quad \forall i\neq \ell, \quad t \in[\bar{t}+\epsilon,t_f].
\end{equation}
As a result, they have no effect on the component $v_{\ell}(t)$ and we obtain
\begin{align}
\dot{v}_{\ell}(t) = u_{1\ell}^*(t) \frac{\partial g_{1}}{\partial x_\ell}(\vec{x}^*(t)) v_{\ell}(t), \quad t \in[\bar{t}+\epsilon,t_f].
\end{align}
This will be of importance in the proof of Theorem \ref{TheoremSI}.

\section{Main Results} \label{mainresults}

\subsection{Scalar Extremum Seeking}

Extremum seeking (ES) offers a systematic approach to address the optimization problem
\begin{equation}
\min F(x), \label{optprob}
\end{equation}
with $x\in\mathbb{R}$, the nonlinear map $F: \mathbb{R}\rightarrow\mathbb{R}$, $F\in \mathcal{C}^1$, without any gradient information being available.
In the following, we consider the class of ES systems introduced by \cite{Duerr13}
\begin{equation}
\dot{x}(t) = g_1(F(x(t)))u_1(t) + g_2(F(x(t)))u_2(t), \; x(0) = x_0, \label{system}
\end{equation}
where $g_1,g_2: \mathbb{R}\rightarrow\mathbb{R}$, $g_1,g_2\in \mathcal{C}^1$. In \cite{Duerr17}, it was discussed that, apart from technical differences, the setups considered in \cite{Krstic00,Tan06} can be represented by this class as well. The dither functions $u_1,u_2: \mathbb{R}\rightarrow\mathbb{R}$ are assumed to be $T$-periodic with $T=\frac{2\pi}{\omega}$. Note that in the following, both $\omega$ and $T$ will be used, although they contain the same information. A typical approach is to choose the dithers
\begin{equation}
u_1(t) = \sqrt{\omega}\sin(\omega t), \quad u_2(t) = \sqrt{\omega}\cos(\omega t), \label{sincos}
\end{equation}
and the vector fields $g_1,g_2$ such that their Lie bracket
\begin{equation}
[g_1,g_2](F) = \frac{\partial g_2(F)}{\partial F}g_1(F)-\frac{\partial g_1(F)}{\partial F}g_2(F) \eqqcolon - g_0(F) \label{LieBr}
\end{equation}
is $1$. Then for $\omega\rightarrow\infty$, the trajectories of \eqref{system} follow those of the gradient flow system $\dot{\bar{x}}=-\frac{1}{2}\nabla F(\bar{x})$ arbitrarily close \cite{Duerr13,Grushkovskaya17}. As a result, $x$ converges to a solution of \eqref{optprob}.

Apparently, $\omega\rightarrow\infty$ can never be achieved in practice. However, also for finite $\omega$, $\eqref{system}$ is observed to move on average along a recovered gradient of $F$. This movement is referred to as learning dynamics (LD). Since the recovered gradient is not exact, the LD generally minimize a function other than $F$, called $L_\omega$. The LD can be carved out from $x(t)$ by neglecting the periodic oscillations induced by the dithers, i.e., by regarding only the system state at $T$-multiples $x(kT), \;k\in\mathbb{N}_+$. Then, the LD of \eqref{system} can be described as a gradient descent recursion on $L_\omega$ with fixed step size
\begin{equation}
x(kT) = x((k-1)T) + \nabla L_\omega (x((k-1)T)). \label{grad_descent}
\end{equation}
As detailed above, the LD have been found to move along $\nabla F$ for $\omega\rightarrow\infty$, whereas for finite $\omega$, the LD were not explicitly quantified so far. The main purpose of this paper is thus to give an explicit description of the gradient descent direction $\nabla L_\omega$, valid for any $\omega$. Thereby, the LD and the function $L_\omega$, that the LD effectively minimize, are characterized.

\begin{figure}[h]
	\centering
	\input{Abbildungen/u1u21Dgen.tex}
	\caption{Trigonometric dithers $u_1(t)$ and $u_2(t)$ which are compliant to \textit{A1}-\textit{A3}. Note that the individual needles in the sampled dither $\bar{u}_2(t)$ form needle pairs of opposite sign (illustrated by matching colors).}
	\label{u1u2}
\end{figure}
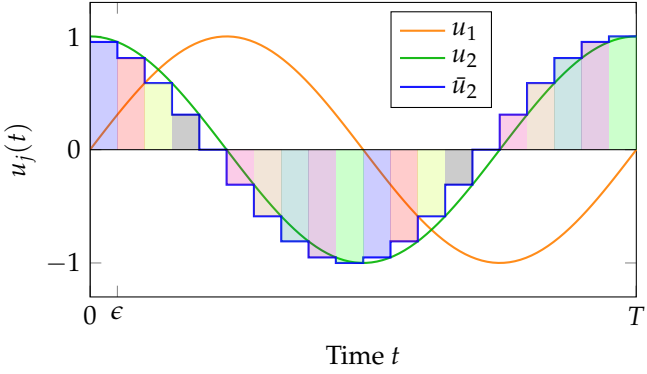

We assume that the following holds for the dither functions:

\begin{enumerate}
	\item[\textit{A1:}] $u_1,u_2$ are piecewise continuous and bounded.
	\item[\textit{A2:}] The function $u_1$ is point-symmetric to $(\frac{T}{2},0)$, i.e., it holds that $u_1(t)=-u_1(T-t)$ for all $t \in [0,T]$.
	\item[\textit{A3:}] For $u_2$ it holds $u_2(t)=-u_2(\frac{T}{2}+t)$ for all $t \in [0,\frac{T}{2}]$.
\end{enumerate}

\begin{remark}
	Note that \textit{A2} and \textit{A3} imply that $u_1$ and $u_2$ have zero mean on $[0,T]$.
\end{remark}

\begin{remark}
	We presume that \textit{A1}-\textit{A3} are mild conditions for dithers commonly considered in ES. For example, they are fulfilled by the well-known trigonometric dithers \eqref{sincos}, but also by square-wave or sawtooth dithers proposed in \cite{Tan08}.
\end{remark}

The idea is now to sample and approximate the dither $u_2$ by needle-shaped functions. We restrict the sampling interval $\epsilon$, i.e., the length of the individual needles, to even divisors of $T$. This means that if the sampling interval is $\epsilon = \frac{T}{2N}, \; N\in \mathbb{N_+}$,
then $u_2(t)$ is approximated by $2N$ needles in the interval $[0,T)$. The sampled dither function is thus
\begin{equation}
\bar{u}_2(t)=u_2(i\epsilon), \; t\in[(i-1)\epsilon, i\epsilon), \; i=1,\ldots,N. \label{u2GU}
\end{equation}
Because of \textit{A3} and the even number of samples, for every needle in the time interval $[0,\frac{T}{2})$, there is a corresponding needle with same amplitude but opposite sign in $[\frac{T}{2},T)$, such that we can extract ``needle pairs'' out of $\bar{u}_2(t)$. This is illustrated in Fig. \ref{u1u2}, where the same-colored areas form needle pairs of opposite sign. This fact will become crucial in order to establish Theorem \ref{TheoremGU}.

The following theorem explicitly relates the solutions of \eqref{system} at times $t=0$ and $t=T$, and thereby quantifies the gradient recovered by ES.

\begin{theorem}
	Suppose that Assumptions \textit{A1}-\textit{A3} on the dither functions $u_1,u_2$ hold. Let $x^*(t)$ denote the solution of \eqref{system}, when $u_1(t)$ and $u_2(t) \equiv 0$ are applied, i.e., $x^*(t)$ fulfills
	\begin{equation}
	\dot{x}^*(t)=g_1(F(x^*(t)))u_1(t), \: x^*(0)=x_0. \label{nomsystem_GU}
	\end{equation}
	Assume that $x^*(t)$ exists on $[0,T]$.	Let $\Phi(t,t_0)$ be the STM corresponding to the time-varying variational equation
	\begin{equation}
	\dot{v}(t)=u_1(t)\frac{\partial g_1}{\partial F}(F(x^*(t)))\frac{\partial F}{\partial x}(x^*(t)) v(t) \label{vareqGU}
	\end{equation}
	with initial time $t_0$, and let $g_1,g_2$ satisfy \eqref{LieBr}. Consider system \eqref{system}, where $u_1(t)$, and $\bar{u}_2(t)$ as in \eqref{u2GU}, are applied. Then
	\begin{align}
	&x(T) = x_0 + \mathcal{O}(T^2) \label{xTNfinGU} \\
	&\hspace{-3pt}+ \epsilon\hspace{-1pt}\sum_{i=1}^{N} u_2(i\epsilon)\hspace{-6pt}\int\displaylimits_{i\epsilon}^{\frac{T}{2}-i\epsilon}\hspace{-2pt} \frac{\partial F}{\partial x}(x^*(\tau)) \Phi(0,\tau) u_1(\tau) g_0(F(x^*(\tau))) \mathrm{d}\tau. \nonumber
	\end{align}
	Moreover,
	\begin{align}
	&\lim_{N \rightarrow \infty} x(T) = x_0 + \mathcal{O}(T^2) \label{xTNinfGU}\\
	&+ \int\displaylimits_{0}^{\frac{T}{2}} u_2(t) \hspace{-4pt} \int\displaylimits_{t}^{\frac{T}{2}-t} \hspace{-1pt} \frac{\partial F}{\partial x}(x^*(\tau)) \Phi(0,\tau) u_1(\tau) g_0(F(x^*(\tau))) \mathrm{d}\tau \mathrm{d}t. \nonumber
	\end{align}
	\label{TheoremGU}
\end{theorem}
The proof can be found in Appendix \ref{AppProof1}. Note that \eqref{xTNinfGU} characterizes the solution of \eqref{system} when $u_2(t)$ is applied, since $\lim_{N\rightarrow\infty}\bar{u}_2(t)=u_2(t)$ due to \textit{A1}. A continuation of \eqref{xTNinfGU} gives a gradient descent recursion of the form \eqref{grad_descent}. Because $x^*(t)$ depends only on its initial condition \mbox{$x^*(T_k)=x(T_k)$} (with $T_k=(k-1)T$), $\nabla L_\omega$ is given by \eqref{xTNinfGU} as
\begin{align}
&\nabla L_\omega (x(T_k))  = \mathcal{O}(T^2) \label{grad_approx} \\
&+ \hspace{-8pt}\int\displaylimits_{T_k}^{T_k+\frac{T}{2}}\hspace{-8pt}u_2(t)\hspace{-6pt}\int\displaylimits_{t}^{\frac{T}{2}-t} \frac{\partial F}{\partial x}(x^*(\tau)) \Phi(T_k,\tau) u_1(\tau) g_0(F(x^*(\tau))) \mathrm{d}\tau \mathrm{d}t. \nonumber
\end{align}
Consequently, Theorem \ref{TheoremGU} gives an approximate quantification of the gradient descent direction $\nabla L_\omega$, and thus describes the LD of system \eqref{system}. This result is independent of the parameter $\omega$ or convexity properties of the function $F$.

The theoretical insight into the mechanics of ES that can be gained from \eqref{xTNinfGU} is extremely valuable. It shows that the LD evolve along a weighted averaged gradient of the function $F$, even for very general dithers $u_1,u_2$. The weighting factors of the gradient are $\Phi(0,\tau)$, $u_1(\tau)$ and $g_0(F(x^*(\tau)))$. Note that for many vector fields $g_1,g_2$ commonly considered in ES (e.g. $g_1=F,g_2=1$ or $g_1=\sin(F), g_2=\cos(F)$), it holds that $g_0(F)=1$ \cite{Grushkovskaya17}, such that this factor even vanishes in the integral. The inner integral (which corresponds to the averaged gradient) is then weighted with the second dither $u_2(t)$ and averaged once again. This  observation gives rise to the interpretation that when choosing the input period $T$ large enough, the LD of the ES system might ``even out'' local minima in $F$. As a result, the LD converge to the global minimum instead of getting stuck in local minima. This phenomenon is indeed observed in practice as seen later. 

\begin{remark}
	Standard ES analysis already indicates that the right-hand side of the "average" ES system corresponds to an averaged version of the static objective function (see e.g. \cite{Krstic00,Tan06}). In contrast to our results, however, this method does not address the \textit{explicit} solution of the \textit{original} ES system.
\end{remark}
\begin{remark}
	Note that $\nabla L_\omega (x(T_k))$ in \eqref{grad_approx} depends on $x^*(t)$, the solution of a nonlinear differential equation. Therefore, it cannot be computed directly and an approximate numerical solution must be obtained instead. This fact indeed does not diminish the valuable insight gained from \eqref{grad_approx}.
\end{remark}

\begin{remark}
	Theorem \ref{TheoremGU} also includes the case of needle-shaped inputs treated in \cite{Michalowsky16}. There are only two needles of length $\epsilon$, such that the rest term in \eqref{xTNfinGU} is estimated to $\mathcal{O}(\epsilon^2)$.
\end{remark}

\subsection{\fontdimen2\font=3pt Multidimensional Extremum Seeking with Sequential Dither}

In this section, the characterization of the LD of ES systems is extended to the multidimensional case. Consider the multidimensional system
\begin{equation}
\dot{\vec{x}}(t) = g_1(F(\vec{x}(t)))\vec{u}_1(t) + g_2(F(\vec{x}(t)))\vec{u}_2(t), \; \vec{x}(0)=\vec{x}_0 \label{system_mult},
\end{equation}
with $\vec{x}(t)\in\mathbb{R}^n$, $F:\mathbb{R}^n\rightarrow\mathbb{R}, F\in\mathcal{C}^1$, \mbox{$g_1,g_2: \mathbb{R}\rightarrow\mathbb{R}$} and the $nT$-periodic dither functions \mbox{$\vec{u}_1,\vec{u}_2: \mathbb{R}\rightarrow\mathbb{R}^n$}. Here, we consider the special multidimensional sequence, where a general scalar dither $u_1,u_2$ is applied sequentially in all dimensions, while the other dithers are zero meanwhile. This sequence, shown in Fig. \ref{u1u2mult}, is defined by
\begin{equation}
u_{ji}(t) = \begin{cases}
u_j(t-(i-1)T) & t\in[(i-1)T,iT) \\
0 & \text{else}
\end{cases}, \; \substack{j=1,2 \\ i=1,\ldots,n}, \label{u_mult}
\end{equation}
with $\vec{u}_j(t)=[u_{j1}(t),\ldots,u_{jn}(t)]^\top, \; j=1,2$. Again, the scalar $u_1,u_2$ need to fulfill \textit{A1}-\textit{A3}.
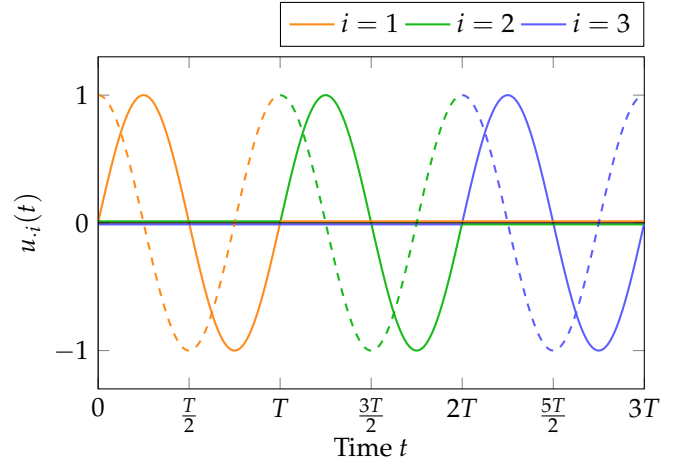
\begin{figure}[h]
	\centering
	\input{Abbildungen/u1lu2lsequential.tex}
	\caption{The multidimensional dither sequences $u_{1i}(t)$ (solid) and $u_{2i}(t)$ (dashed) for $n=3$. The scalar dithers $u_1,u_2$ are merely applied sequentially in all dimensions.}
	\label{u1u2mult}
\end{figure} \\
The following theorem characterizes the LD of system \eqref{system_mult} when treated with this sequential dither.

\begin{theorem}
	Suppose that Assumptions \textit{A1}-\textit{A3} on the scalar dithers $u_1,u_2$ hold. Let $\ell\in\{1,\ldots,n\}$. Denote $\vec{x}^*(t)$ the solution of \eqref{system_mult} when $\vec{u}_1(t)$ is as defined in \eqref{u_mult} and \mbox{$\vec{u}_2(t) \equiv \vec{0}$}, and assume that $\vec{x}^*(t)$ exists on $[0,\ell T]$. Let \mbox{$\Phi_i(t,t_0),\Phi_i:\mathbb{R}\times\mathbb{R}\rightarrow\mathbb{R}$} be the STM corresponding to the time-varying variational equation
	\begin{equation}
	\dot{v}_i(t) = u_{1i}(t)\frac{\partial g_1}{\partial F}(F(\vec{x}^*(t)))\frac{\partial F}{\partial x_i}(\vec{x}^*(t)) v_i(t) \label{vareqSI}
	\end{equation}
	with initial time $t_0$, and let $g_1,g_2$ satisfy \eqref{LieBr}. Consider the multidimensional system \eqref{system_mult}, where $\vec{u}_1(t)$ and $\vec{u}_2(t)$ are applied as defined in \eqref{u_mult}. Then,
	\begin{align}
	&\vec{x}(\ell T) = \vec{x}_0 + \mathcal{O}(\ell^2 T^2) \label{xnTSI} \\
	&+ \begin{bmatrix}
	\hspace{-57pt}\int\displaylimits_{0}^{\frac{T}{2}} u_2(t) \int\displaylimits_{t}^{\frac{T}{2}-t}\Big( \frac{\partial F}{\partial x_1}(\vec{x}^*(\tau)) \Phi_1(0,\tau) \\
	\hspace{100pt} \cdot u_1(\tau) g_0(F(\vec{x}^*(\tau))) \Big) \mathrm{d}\tau \mathrm{d}t \\
	\vdots \\
	\hspace{-10pt}\int\displaylimits_{0}^{\frac{T}{2}} u_2(t) \int\displaylimits_{t}^{\frac{T}{2}-t} \Big( \frac{\partial F}{\partial x_\ell}(\vec{x}^*(T_\ell+\tau)) \Phi_\ell(T_\ell,T_\ell+\tau) \\
	\hspace{79pt} \cdot u_1(\tau) g_0(F(\vec{x}^*(T_\ell+\tau))) \Big) \mathrm{d}\tau \mathrm{d}t \\
	\vec{0}_{n-\ell}
	\end{bmatrix}\hspace{-2pt}. \nonumber
	\end{align}
	\label{TheoremSI}
\end{theorem}

\begin{sketchofproof}
	The proof of Theorem \ref{TheoremSI} follows closely the lines of the proof of Theorem \ref{TheoremGU}. The main idea is again to sample the input functions in dimension $i$ by an even number of needles in the interval $[(i-1)T,iT)$. Consider $\vec{\Phi}(t,t_0)$, the STM corresponding to
	\begin{equation}
	\vec{\dot{v}}(t) = \vec{u}_1(t)\frac{\partial g_1}{\partial F}(F(\vec{x}^*(t))) \nabla F(\vec{x}^*(t))^\top \vec{v}(t).
	\end{equation}
	With the fact that the needles are applied in one dimension at a time, and the principles of variational calculus given in Section \ref{VariationalCalculus}, one can express the solution of \eqref{system_mult} after $\ell T$ as
	\begin{align}
	&\vec{x}(\ell T) = \vec{x}_0 + \mathcal{O}(\ell T^2) \label{xnTSI1}\\
	&+ \epsilon \sum_{i=1}^{\ell} \sum_{j=1}^{2N} u_2(j\epsilon) \vec{\Phi}(\ell T,T_{i}+j\epsilon)
	\begin{bmatrix}
	\vec{0}_{i-1} \\ g_2(\vec{x}^*(T_{i}+j\epsilon)) \\ \vec{0}_{n-i}
	\end{bmatrix}. \nonumber
	\end{align}
	Because of the point-symmetries in $\vec{u}_1(t)$, it holds that $\vec{\Phi}(T_{i+1},T_i)=\vec{I}$ for all $i$. With the semi-group property of the STM, and the results presented in Section \ref{SpecPertInp}, one can replace $\vec{\Phi}(\ell T,T_{i}+j\epsilon)$ in \eqref{xnTSI1} by the scalar $\Phi_i(T_{i+1},T_{i}+j\epsilon)$. Again, we apply the symmetry property of $u_2(t)$ and write the $\mathcal{O}(\epsilon)$ terms as an integral. Letting $\epsilon\rightarrow 0$ proves \eqref{xnTSI}.
\end{sketchofproof}

Theorem \ref{TheoremSI} shows that using the special sequential sequence as defined in \eqref{u_mult}, a component-wise and decoupled weighting and averaging of the gradient is performed. Furthermore, formula \eqref{xnTSI} reveals that the LD move primarily along the dimension where the scalar input was applied (as the integral term indicates). However, due to the non-vanishing remainder, the system moves slightly along the other dimensions as well. 

\section{Numerical Evaluation} \label{simulations}

In this section, we compare the simulated LD $x(kT)_\mathrm{Sim}$ of an ES system with the result of the recursion \eqref{grad_descent} and \eqref{grad_approx}, denoted by $x(kT)_\mathrm{Rec}$. We consider the system from \cite{Duerr13}
\begin{equation}
\dot{x}(t) = F(x(t))\sqrt{\omega}\sin(\omega t) - a\sqrt{\omega}\cos(\omega t), \label{system_example}
\end{equation}
with initial condition $x(0)=1.8$. The parameter $\omega=\frac{2\pi}{T}$ will be adapted to display its various effects.
\begin{figure}[h]
	\centering
	\input{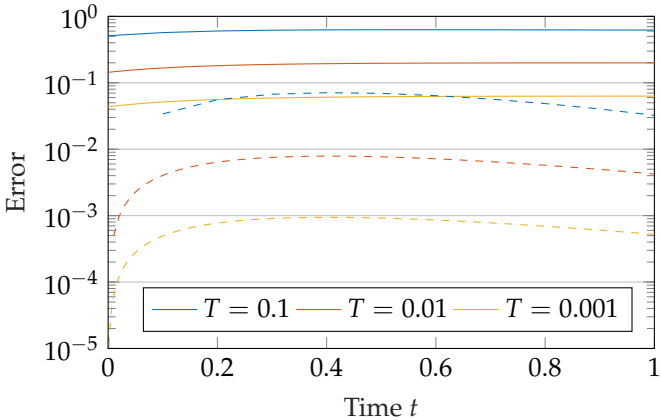}
	\caption{The absolute error between the simulated LD and the recursion (dashed), and the size of the tube around the LD (solid), with $F_1(x)$, $a=5$.}
	\label{error}
\end{figure}
\begin{example}
	Consider \eqref{system_example} together with the quadratic function $F_1(x)=\frac{1}{2}x^2$. Fig. \ref{error} shows the absolute error $|x(kT)_\mathrm{Rec}-x(kT)_\mathrm{Sim}|$ between the simulated LD and the recursion for system \eqref{system_example} for various $T$. The chosen period times differ by a factor of $10$, as do the errors at a certain time point $t$. This observation matches the theory very well, e.g., consider $T_1=10 T_2$. Then using the larger period time $T_1$ leads to an error at time $T_1$ of $\mathcal{O}(T_1^2)$, whereas performing the recursion $10$ times with $T_2$ causes an error at time $T_1$ of $\mathcal{O}(10T_2^2)=\mathcal{O}(0.1T_1^2)$. In \cite{Duerr13}, the LD could be verified to move in a tube around the gradient flow system. Observe from Fig. \ref{error} that the recursion gives a much more accurate estimation of the ES system's LD.
\end{example}
\begin{figure}[h]
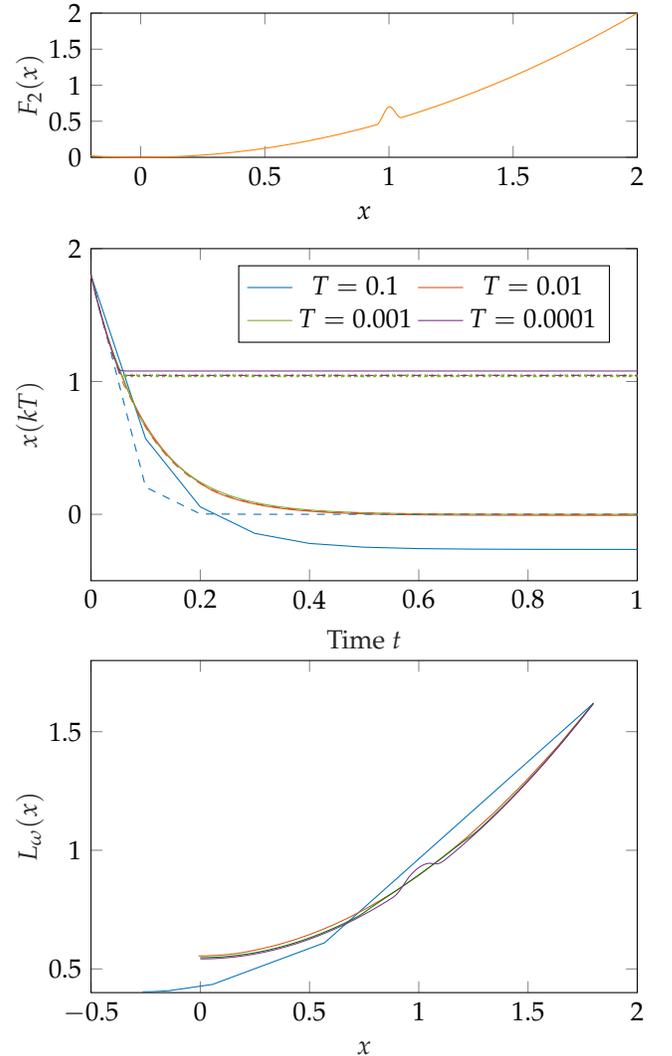

	\centering
	\input{Abbildungen/F.tex}
	\input{Abbildungen/xkT1D.tex}
	\input{Abbildungen/Lomega.tex}
	\caption{The non-convex test function $F_2(x)$ (top), the simulated LD $x(kT)_\mathrm{Sim}$ (solid) and the recursion $x(kT)_\mathrm{Rec}$ (dashed) with $F_2(x)$ (middle), and the simulated $L_\omega(x)$ corresponding to $F_2(x)$ (bottom), $a=20$.}
	\label{xkT1D}
\end{figure}
\begin{example}
	This example shall demonstrate the recursion's ability to represent the simulated LD with non-convex functions. Consider the function $F_2(x)$ as depicted in Fig. \ref{xkT1D} (top). It has a sharp local minimum between $x_0$ and its global minimum at $x=0$, and is a quadratic function elsewhere. In Fig. \ref{xkT1D} (middle), the simulated LD and the recursion are depicted. For $T=0.1,0.01$, both the simulation and the recursion converge globally, and for $T=0.0001$, both get stuck in the local minimum. However, for $T=0.001$, the simulation passes through the local minimum while the recursion gets stuck. We can infer that in borderline cases, when the ES system barely converges globally, the recursion's informative value is limited due to its intrinsic uncertainty. Nonetheless, for clearer cases, the recursion displays the actual LD very well. 
	
	Fig. \ref{xkT1D} (bottom) shows $L_\omega(x)$ generated by the simulated $x(kT)_\mathrm{Sim}$, where \eqref{system_example} was treated with $F_2(x)$. The function $L_\omega(x)$ was numerically integrated using the Euler method from \eqref{grad_descent} with a normalized scaling. This example illustrates the property that $L_\omega(x)$ can become convex for certain $\omega$ although $F_2(x)$ is not. It can be observed that for higher $T$, $L_\omega(x)$ is a convex function, whereas for small $T$, it shows a local minimum similar to $F_2(x)$. Consequently, when starting to the right of this local minimum, $x(kT)_\mathrm{Sim}$ does not converge near the global minimizer of $F_2(x)$.
\end{example}
\begin{figure}[h]
	\centering
	\input{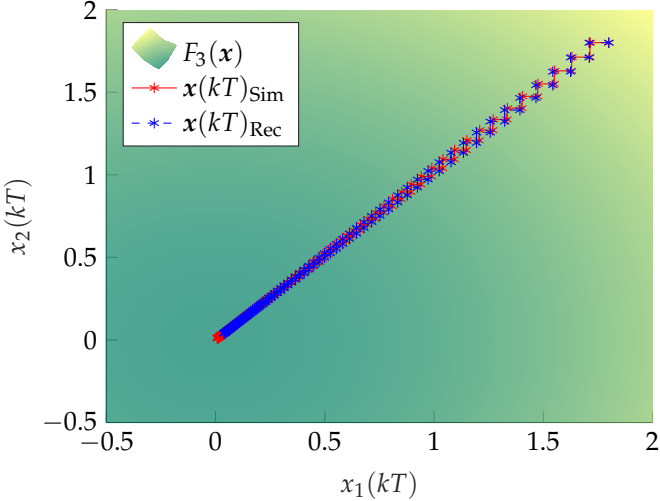}
	\caption{The two-dimensional LD using $F_3(\vec{x})$ and $T=0.01$. The value of $F_3(\vec{x})$ is color-coded from yellow (high) to green (low), $a=10$.}
	\label{xkT2D}
\end{figure}
\begin{example}
	The third example is devoted to the two-dimensional extension of \eqref{system_example}, where the one-dimensional sine-waves are applied sequentially according to \eqref{u_mult}. The initial condition $\vec{x}(0) = [1.8,1.8]^\top$ and the convex function $F_3(\vec{x})=\frac{1}{2}(x_1^2+x_2^2)$ are used.
	In Fig. \ref{xkT2D}, simulation results of this setup are depicted. The staircase-shaped LD, that were predicted by \eqref{xnTSI}, are clearly visible. Over the first few time periods, the recursion gives a very precise approximation of the LD. However, the approximation error sums up such that the recursion becomes less and less precise over time. Note that this observation does not disagree with our main results. Nonetheless, the recursion converges near the minimizer at the origin, just like the simulated LD.
\end{example}

\section{Summary} \label{conclusion}

In this paper, we gave an explicit recursion for the LD of scalar ES systems with static maps. This recursion approximately quantifies the gradient information recovered by ES, and reveals that it corresponds to an averaged gradient of the objective function. As this property holds without strong restrictions on the objective function, the recursion is also able to represent convergence of the LD to the global minimum, despite the presence of local minima. Furthermore, we presented a special multidimensional dither sequence and showed that an ES system, treated with this sequence, moves along an averaged gradient as well. Eventually, we illustrated and verified our results in simulations. Since in general, dynamic maps are considered in ES, an extension of the presented analysis to this case seems worthwhile.

\section{Appendix}

\subsection{Proof of Theorem \ref{TheoremGU}} \label{AppProof1}

\begin{proof}

The proof relies on the principles of variational calculus presented in Section \ref{methods}. We define
\begin{equation}
\bar{u}_2^{(j)}(t) = \begin{cases} \bar{u}_2(t) & t\in[0,(j-1)\epsilon) \\
0 & t\notin [0,(j-1)\epsilon)
\end{cases}, \; j=1,\ldots,2N+1,
\end{equation}
consisting of the first $(j-1)$ needles and denote
\begin{itemize}
	\item  $x^{(j)}(t)$ the solution of \eqref{system} when $u_1(t)$ and $\bar{u}_2^{(j)}(t)$ are applied,
	\item  $v^{(j)}(t)$ the variational variable which describes the variation of $x^{(j+1)}(t)$ from $x^{(j)}(t)$.
\end{itemize}
Consider system \eqref{system}, where $\bar{u}_2^{(j)}(t)$ is applied, and its solution $x^{(j)}(t)$. Now apply $\bar{u}_2^{(j+1)}(t)$. Then, $\bar{u}_2^{(j)}(t)$ corresponds to the nominal input that is perturbed on \mbox{$t\in[(j-1)\epsilon,j\epsilon)$}, and $x^{(j)}(t)$ to the nominal trajectory, such that
\begin{equation}
x^{(j+1)}(t)=x^{(j)}(t)+\epsilon v^{(j)}(t)+\mathcal{O}(\epsilon^2). \label{xvarproof}
\end{equation}
The variational variable $v^{(j)}(t)$ evolves as
\begin{equation}
\dot{v}^{(j)}(t) = u_1(t)\frac{\partial g_1}{\partial F}(F(x^{(j)}(t)))\frac{\partial F}{\partial x}(x^{(j)}(t)) v^{(j)}(t) \label{vareqproof}
\end{equation}
and has initial condition \mbox{$v^{(j)}(j\epsilon) = g_2(F(x^{(j)}(j\epsilon)))u_2(j\epsilon)$}, since $\bar{u}_2^{(j)}(j\epsilon)=0$.

The proof is subdivided into four sections. In the first part, we will derive a representation of the solution of \eqref{system} after one dither period $T$ via variational calculus, where the terms linear in $\epsilon$ will be comprised of STMs. Second, we will derive a useful symmetric property of these STMs. Using this, the terms of order $\mathcal{O}(\epsilon)$ will be expressed as an integral to prove \eqref{xTNfinGU} in the third part. In the last part, we let the length of the needles tend to zero to show \eqref{xTNinfGU}.

\textbf{1):} We begin to show by induction that when $\bar{u}_2^{(j+1)}(t)$ is applied to the system, the solution of \eqref{system} is
\begin{align}
x^{(j+1)}(t) &= x^*(t) + \epsilon \sum_{i=1}^{j} \Phi(t,i\epsilon) g_2(F(x^*(i\epsilon)))u_2(i\epsilon) \nonumber \\
&+\sum_{i=1}^{j} \Big(R_i(\epsilon)+S_i(\epsilon)\Big), \label{xjinductionGU}
\end{align}
for $t\in[j\epsilon,T]$, where each of the remainders $R_i(\epsilon)=\mathcal{O}(\epsilon^2)$ and each $S_i(\epsilon)=\mathcal{O}(2(i-1)\epsilon^2)$.

\textit{Step 1: Basis} By the principles of variational calculus presented in Section \ref{methods}, applying $\bar{u}_2^{(2)}(t)$, we obtain for the solution of \eqref{system}
\begin{equation}
x^{(2)}(t) = \underbrace{x^{(1)}(t)}_{=x^*(t)} + \epsilon v^{(1)}(t) + R_1(\epsilon), 
\end{equation}
for $t\in [\epsilon,T]$. For the remainder it holds $R_1(\epsilon)=\mathcal{O}(\epsilon^2)$. As $x^{(1)}(t)=x^*(t)$, the variational variable $v^{(1)}(t)$ fulfills \eqref{vareqGU} with initial condition $v^{(1)}(\epsilon)=g_2(F(x^*(\epsilon)))u_2(\epsilon)$. Since $\Phi(t,t_0)$ is the STM of \eqref{vareqGU}, we can express $x^{(2)}(t)$ by
\begin{equation}
x^{(2)}(t) = x^*(t) + \epsilon\Phi(t,\epsilon)g_2(F(x^*(\epsilon)))u_2(\epsilon) + R_1(\epsilon), \label{xepsilon}
\end{equation}
for $t\in [\epsilon,T]$, which is \eqref{xjinductionGU} for $j=2$.

\textit{Step 2: Inductive Step} Assume that \eqref{xjinductionGU} holds when $\bar{u}_2^{(j)}(t)$ is applied, i.e., that the solution $x^{(j)}(t)$ reads
\begin{align}
x^{(j)}(t) &= x^*(t) + \epsilon \sum_{i=1}^{j-1} \Phi(t,i\epsilon) g_2(F(x^*(i\epsilon)))u_2(i\epsilon) \nonumber \\
&+\sum_{i=1}^{j-1}\Big(R_i(\epsilon)+S_i(\epsilon)\Big), \label{xjinductionGU2}
\end{align}
for $t\in[(j-1)\epsilon,T]$. Consider $v^{(j)*}(t)$, which fulfills \eqref{vareqGU} and has initial conditions $v^{(j)*}(j\epsilon) = g_2(F(x^*(j\epsilon)))u_2(j\epsilon)$. Next, we apply $\bar{u}_2^{(j+1)}(t)$. With Lemma \ref{LemmaMultNeed} (see Appendix \ref{AppLemma1}), equation \eqref{xvarproof} becomes 
\begin{align}
&x^{(j+1)}(t) = x^{(j)}(t) + \epsilon v^{(j)*}(t) + \underbrace{\epsilon v_R^{(j)}(\epsilon)}_{\eqqcolon S_j(\epsilon)} + R_j(\epsilon) \nonumber \\
&\stackrel{\eqref{xjinductionGU2}}{=} x^*(t) + \epsilon \sum_{i=1}^{j-1} \Phi(t,i\epsilon) g_2(F(x^*(i\epsilon)))u_2(i\epsilon)+\epsilon v^{(j)*}(t) \nonumber \\
&+\sum_{i=1}^{j-1}\Big(R_i(\epsilon)+S_i(\epsilon)\Big)+R_j(\epsilon)+S_j(\epsilon),
\end{align}
for $t \in[j\epsilon,T]$. Again, it holds that $R_j(\epsilon)=\mathcal{O}(\epsilon^2)$. Note that $x^{(j)}(t) = x^*(t) + \mathcal{O}((j-1)\epsilon)$, because according to \eqref{xjinductionGU}, each of the summands in the first sum is $\mathcal{O}(\epsilon)$. With Lemma \ref{LemmaMultNeed}, it holds that \mbox{$v_R^{(j)}(\epsilon)=\mathcal{O}(2(j-1)\epsilon)$} and thus \mbox{$S_j(\epsilon) = \epsilon v_R^{(j)}(\epsilon)=\mathcal{O}(2(j-1)\epsilon^2)$}. Using the STM of \eqref{vareqGU}, we obtain
\begin{align}
x^{(j+1)}(t) &= x^*(t) + \epsilon \sum_{i=1}^{j}  \Phi(t,i\epsilon) g_2(F(x^*(i\epsilon)))u_2(i\epsilon) \nonumber \\
&+\sum_{i=1}^{j}\Big(R_i(\epsilon)+S_i(\epsilon)\Big),
\end{align}
for $t\in[j\epsilon,T]$, which is exactly \eqref{xjinductionGU}.

Using \eqref{xjinductionGU}, we can express the solution of \eqref{system} at $t=T$, when $\bar{u}_2^{(2N+1)}(t)=\bar{u}_2(t)$ was applied to the system, as
\begin{align}
x(T) &= x^{(2N+1)}(T) \nonumber \\
&= x^*(T) + \epsilon \sum_{i=1}^{2N} \Phi(T,i\epsilon) g_2(F(x^*(i\epsilon)))u_2(i\epsilon) \nonumber \\
&+\sum_{i=1}^{2N}\Big(R_i(\epsilon)+S_i(\epsilon)\Big). \label{xTGU1}
\end{align}

In the next step, we will give a precise assessment of the remainder terms in \eqref{xTGU1}. We have established that the individual terms can be estimated as $R_i(\epsilon)=\mathcal{O}(\epsilon^2)$ and \mbox{$S_i(\epsilon)=\mathcal{O}(2(i-1)\epsilon^2)$}. Taking the summation of the remainder terms into account, the overall remainder is assessed as
\begin{align}
\sum_{i=1}^{2N} \Big(R_i(\epsilon) + S_i(\epsilon)\Big) &= \mathcal{O}(\big(2N+\sum_{i=1}^{2N}2(i-1)\big)\epsilon^2) \nonumber \\
&= \mathcal{O}((2N)^2\epsilon^2) \stackrel{\epsilon = \frac{T}{2N}}{=} \mathcal{O}(T^2).
\end{align}

Recall that because of \textit{A3} and the even number of needles, the sampled dither $\bar{u}_2(t)$ is comprised of ``needle pairs'' of same amplitude, but opposite sign. In the following, we will formalize this idea and carve out its effect on $x(T)$. First we split the sum in \eqref{xTGU1}
\begin{align}
x(T) &= x^*(T) + \epsilon \sum_{i=1}^{N} \Phi(T,i\epsilon) g_2(F(x^*(i\epsilon)))u_2(i\epsilon) \nonumber \\
&+\epsilon \sum_{i=N}^{2N} \Phi(T,i\epsilon) g_2(F(x^*(i\epsilon)))u_2(i\epsilon)+\mathcal{O}(T^2),
\end{align}
then use that $u_2(i\epsilon)=-u_2(i\epsilon-\frac{T}{2})$ in the second sum due to \textit{A3}, and perform an index shift to obtain
\begin{align}
x(T) &= x^*(T) + \epsilon \sum_{i=1}^{N} u_2(i\epsilon)\Big( \Phi(T,i\epsilon) g_2(F(x^*(i\epsilon))) \nonumber \\
&- \Phi(T,\tfrac{T}{2}+i\epsilon) g_2(F(x^*(\tfrac{T}{2}+i\epsilon))) \Big) +\mathcal{O}(T^2). \label{xTGU2}
\end{align}

\textbf{2):} The second part of the proof is dedicated to finding a simpler form of the STMs occurring in \eqref{xTGU2}. Because of \textit{A2}, it holds for the dither that $u_1(t)=-u_1(T-t)$. As a consequence, the nominal solution goes along the vector field $g_1(F(x^*(t)))u_1(t)$ in $[0,\frac{T}{2}]$, and back along the same vector field in $[\frac{T}{2},T]$. Then the nominal solution fulfills
\begin{equation}
x^*(t)=x^*(T-t). \label{xsymmetryGU}
\end{equation}
Using this symmetry property and $u_1(t)=-u_1(T-t)$ on the definition of $\Phi(t,t_0)$ yields
\begin{align}
&\Phi(t,t_0) =\exp\left( \int_{t_0}^{t} u_1(\tau)\frac{\partial g_1}{\partial F}(F(x^*(\tau))) \frac{\partial F}{\partial x}(x^*(\tau)) \mathrm{d}\tau \right) \nonumber \\
&\stackrel{\phantom{s\coloneqq T-\tau}}{=} \exp\left( \int_{T-t_0}^{T-t} u_1(s)\frac{\partial g_1}{\partial F}(F(x^*(s))) \frac{\partial F}{\partial x}(x^*(s)) \mathrm{d}s \right) \nonumber \\
&\stackrel{\phantom{s\coloneqq T-\tau}}{=} \Phi(T-t,T-t_0). \label{PhisymmetryGU}
\end{align}
The relation \eqref{PhisymmetryGU} and the semi-group property of the STM is used to establish that, if $0\le i\epsilon\le\frac{T}{2}$,
\begin{align}
\Phi(T,i\epsilon)&\stackrel{\phantom{\eqref{PhisymmetryGU}}}{=}\Phi(T,\tfrac{T}{2})\Phi(\tfrac{T}{2},i\epsilon) \nonumber \\
&\stackrel{\eqref{PhisymmetryGU}}{=} \Phi(0,\tfrac{T}{2})\Phi(\tfrac{T}{2},i\epsilon) = \Phi(0,i\epsilon) \label{PhiGU1}
\end{align}
and
\begin{equation}
\Phi(T,\tfrac{T}{2}+i\epsilon)\stackrel{\eqref{PhisymmetryGU}}{=} \Phi(0,\tfrac{T}{2}-i\epsilon). \label{PhiGU2}
\end{equation}

\textbf{3):} In the third section of the proof, we combine the results of the first two sections to come up with \eqref{xTNfinGU}. We use the identities \eqref{xsymmetryGU}, \eqref{PhiGU1} and \eqref{PhiGU2} in \eqref{xTGU2} to write
\begin{align}
&x(T) = \underbrace{x^*(T)}_{\stackrel{\eqref{xsymmetryGU}}{=}x^*(0)=x_0}+\mathcal{O}(T^2)  \nonumber \\
&+ \epsilon \sum_{i=1}^{N} u_2(i\epsilon)\bigg( \underbrace{\Phi(T,i\epsilon)}_{\stackrel{\eqref{PhiGU1}}{=}\Phi(0,i\epsilon)} g_2(F(x^*(i\epsilon))) \nonumber \\
&\hspace{67pt}- \underbrace{\Phi(T,\tfrac{T}{2}+i\epsilon)}_{\stackrel{\eqref{PhiGU2}}{=}\Phi(0,\tfrac{T}{2}-i\epsilon)} g_2(F(\underbrace{x^*(\tfrac{T}{2}+i\epsilon)}_{\stackrel{\eqref{xsymmetryGU}}{=}x^*(\tfrac{T}{2}-i\epsilon)})) \bigg) \nonumber \\
&= x_0 + \mathcal{O}(T^2) \label{xTGU4} \\
&- \epsilon\sum_{i=1}^{N} u_2(i\epsilon)\int\displaylimits_{i\epsilon}^{\frac{T}{2}-i\epsilon} \frac{\mathrm{d}}{\mathrm{d} \tau}\bigg(\Phi(0,\tau)g_2(F(x^*(\tau)))\bigg)\mathrm{d}\tau. \nonumber
\end{align}
Note that writing the term linear in $\epsilon$ as an integral was only possible as there existed ``needle pairs'' of opposite sign. Next, we concentrate on simplifying the integral. Using the product rule, the differentiation property of the STM, and the chain rule, the integral becomes
\begin{align}
&\hspace{-3pt}\int\displaylimits_{i\epsilon}^{\frac{T}{2}-i\epsilon} \hspace{-6pt}  \Phi(0,\tau) \frac{\mathrm{d} g_2}{\mathrm{d} \tau}(F(x^*(\tau))) \hspace{-2pt} + \hspace{-2pt} \frac{\mathrm{d}}{\mathrm{d} \tau}\big(\Phi(0,\tau)\big)g_2(F(x^*(\tau))) \mathrm{d}\tau \nonumber \\
&=\int\displaylimits_{i\epsilon}^{\frac{T}{2}-i\epsilon} \Phi(0,\tau) \frac{\partial g_2}{\partial F}(F(x^*(\tau)))\frac{\partial F}{\partial x}(x^*(\tau))\dot{x}^*(\tau)  \\
& \hspace{0pt} -\Phi(0,\hspace{-1pt}\tau) u_1(\tau)\hspace{-1pt} \frac{\partial g_1}{\partial F}\hspace{-1pt}(F(x^*(\tau)\hspace{-1pt})\hspace{-1pt})\frac{\partial F}{\partial x}(x^*(\tau)\hspace{-1pt}) g_2(F(x^*(\tau)\hspace{-1pt})\hspace{-1pt})  \mathrm{d}\tau. \nonumber 
\end{align}
The term $\dot{x}^*(\tau)$ is the right-hand side of the nominal differential equation \eqref{nomsystem_GU}. Using this, we can write the integral
\begin{equation}
\int\displaylimits_{\epsilon}^{\frac{T}{2}-\epsilon} \frac{\partial F}{\partial x}(x^*(\tau))\Phi(0,\tau)u_1(\tau)\underbrace{\bigg(\frac{\partial g_2}{\partial F}g_1-\frac{\partial g_1}{\partial F}g_2\bigg)}_{\stackrel{\eqref{LieBr}}{=}-g_0(F(x^*(\tau)))} \mathrm{d}\tau, \label{int_final}
\end{equation}
omitting some arguments. Using \eqref{int_final} in \eqref{xTGU4} proves \eqref{xTNfinGU}.

\textbf{4):} In the fourth part of the proof, we perform the limit process of letting $N$ tend to infinity. We define
\begin{equation}
t_i \coloneqq i\epsilon
\end{equation}
to express the limit of the first-order term in \eqref{xTNfinGU} as
\begin{align}
\lim_{N\rightarrow \infty} \sum_{i=1}^{N} &\bigg( (t_i-t_{i-1}) u_2(t_i) \label{riemann}\\
&\cdot\int\displaylimits_{t_i}^{\frac{T}{2}-t_i} \frac{\partial F}{\partial x}(x^*(\tau)) \Phi(0,\tau) u_1(\tau) g_0(F(x^*(\tau))) \mathrm{d}\tau \bigg). \nonumber
\end{align}
Note that as $u_2$ is continuous and bounded due to \textit{A1}, this is the limit of a Riemann sum \cite{Trench03} with partition $\{t_i\}$, which converges to the Riemann integral. Therefore, the first-order term \eqref{riemann} becomes
\begin{equation}
\int\displaylimits_{0}^{\frac{T}{2}} u_2(t) \int\displaylimits_{t}^{\frac{T}{2}-t} \frac{\partial F}{\partial x}(x^*(\tau)) \Phi(0,\tau) u_1(\tau) g_0(F(x^*(\tau))) \mathrm{d}\tau \mathrm{d}t,
\end{equation}
which gives \eqref{xTNinfGU} and completes the proof.

\end{proof}

\subsection{Auxiliary Lemma \ref{LemmaMultNeed}} \label{AppLemma1}

\begin{lemma}
	Consider $v^{(j)}(t)$ and $v^{(j)*}(t)$ as defined in the proof of Theorem \ref{TheoremGU}. Suppose that \mbox{$x^{(j)}(t)=x^*(t)+\mathcal{O}(i\epsilon)$}. Then,
	\begin{equation}
	v^{(j)}(t) = v^{(j)*}(t) + v^{(j)}_R(\epsilon),
	\end{equation}
	where it holds for the remainder $v^{(j)}_R(\epsilon)=\mathcal{O}(2i\epsilon)$. \label{LemmaMultNeed}
\end{lemma}

\begin{proof}
	We perform a Taylor expansion of the system matrix of \eqref{vareqproof}, and the initial condition of $v^{(j)}(t)$ about $x^*$, where we use that $x^{(j)}=x^*+\mathcal{O}(i\epsilon)$. This yields
	\begin{align}
	A(x^{(j)},t) &= A(x^*+\mathcal{O}(i\epsilon),t) = A(x^*,t) + \mathcal{O}(i\epsilon), \\
	g_2(x^{(j)}) &= g_2(x^*+\mathcal{O}(i\epsilon)) = g_2(x^*) + \mathcal{O}(i\epsilon).
	\end{align}
	As a result, the solution of the variational equation is
	\begin{align}
	v^{(j)}(t) = \exp\bigg(\int_{t_0}^t &A(x^*(\tau),\tau) + \mathcal{O}(i\epsilon) \mathrm{d}\tau \bigg) \nonumber\\
	&\cdot\Big(g_2(x^*(j\epsilon))+ \mathcal{O}(i\epsilon)\Big) u_2(j\epsilon).
	\end{align}
	Using the Taylor expansion for $\exp\left(\int_{t_0}^t \mathcal{O}(i\epsilon) \mathrm{d}\tau \right)$ gives
	\begin{align}
	&v^{(j)}(t) = \exp\left(\int_{t_0}^t A(x^*(\tau),\tau)\mathrm{d}\tau \right)(1+\mathcal{O}(i\epsilon)) \nonumber\\
	&\hspace{80pt} \cdot\Big(g_2(x^*(j\epsilon))+\mathcal{O}(i\epsilon)\Big) u_2(j\epsilon) \nonumber \\
	&=\exp\left(\int_{t_0}^t A(x^*(\tau),\tau)\mathrm{d}\tau\right) g_2(x^*(j\epsilon))u_2(j\epsilon) + \mathcal{O}(2i\epsilon).
	\end{align}
	We note that the first term corresponds to $v^{(j)*}(t)$ and the second remainder term fulfills $v^{(j)}_R(\epsilon)=\mathcal{O}(2i\epsilon)$.
\end{proof}

\bibliographystyle{IEEEtran}
\bibliography{BibliographyMA}

\end{document}

%% file: customCommands.tex
\newcommand{\oprocendsymbol}{{$\bullet$}} 					
\def\oprocend {{
\parfillskip=0pt        
\widowpenalty=10000     
\displaywidowpenalty=10000  
\finalhyphendemerits=0  
\leavevmode             
\unskip                 
\nobreak                
\hfil                   
\penalty50              
\hskip.45em              
\null                   
\hfill                  
\oprocendsymbol
\par}}                  

%% file: Abbildungen/needleVariation.tex
\begin{tikzpicture}[>=latex]
	\draw[->] (0,0) -- (2.5,0);
	\draw[->] (0,0) -- (0,2.6);
	\draw[very thick,green!60!black] (-0.1,0.01) -- (1,0.01);
	\draw[green!60!black,dashed] (1,0.01) --  (1,2.3);
	\draw[very thick,green!60!black] (1,2.3) -- (1.3,2.3);
	\draw[green!60!black,dashed] (1.3,2.3) -- (1.3,0.01);
	\draw[very thick,green!60!black] (1.3,0.01) -- (2.3,0.01); 
	\draw[] (-0.1,2.3) -- node[pos=1,anchor=east] {$\alpha_i$} (0,2.3);
	\node[anchor=west,green!60!black] at (1.3,2.2) {$ u_{2i}(t) $};
	\draw[dashed,ultra thick,orange] (-0.1,0.01) -- (2.3,0.01);
	\node[orange,anchor=south] at (2,0.01) {$ u_{2i}^*(t) $};
	\draw[] (1,0.1) -- node[pos=1,anchor=north east,xshift=0.4em] {$\bar{t}$} (1,-0.1);
	\draw[] (1.3,0.1) -- node[pos=1,anchor=north,xshift=0.6em] {$\bar{t}+\epsilon$} (1.3,-0.1);
	\node[name=x0,draw,circle,fill=black,inner sep=0pt,minimum width=4pt,label=135:$\vec{x}_0$] at (3.8,0.3) {};
	\path[] (x0) to [out=-10,in=-155] node[draw=black,circle,fill=black,inner sep=0pt,minimum width=4pt,pos=1,name=startNeedle] {} ++(1.5,0.7) {};
	\node[anchor=south east] at (startNeedle) {$\vec{x}^*(\bar{t})$};
	\draw[green!60!black,thick] (x0) to [out=-10,in=-155] (startNeedle.center) 
	                                 to [bend right=15] node[pos=1,fill=black,name=endNeedle,draw=black,circle,inner sep=0pt,minimum width=4pt] {} ++(0.3,0.7) 
	                                 to [out=180-157,in=-173] node[anchor=south west] {$\vec{x}(t)$} ++(2,0.3) ;
	\draw[green!60!black,thick] (x0) to [out=-10,in=-155] (startNeedle.center) 
	                                 to [bend right=15] ++(0.3,0.7) 
	                                 to [out=180-157,in=-173] node[anchor=south west] {$\vec{x}(t)$} ++(2,0.3) ;
	\node[anchor=south east] at (endNeedle) {$\vec{x}(\bar{t}+\epsilon)$};
	\draw[orange,dashed,thick] (x0) to [out=-10,in=-155] (startNeedle)  
	                                to [out=180-155,in=-170] node[anchor=south west] {$\vec{x}^*(t)$} ++(2.5,-0.3);
	
\end{tikzpicture}

%% file: Abbildungen/u1u21Dgen.tex
\def\T{3.1416}				
\colorlet{u1Color}{orange}			
\begin{tikzpicture}[>=latex]
	\begin{axis}[
			width=\columnwidth,
			height=5.5cm,
			xmin=0, 
			xmax=6.2832,
			ymin=-1.3, 
			ymax=1.3,
			disabledatascaling,
			xlabel=Time $t$,
			ylabel=$u_j (t)$,
			xtick={0,0.3142,6.28},
			xticklabels={ $0$,$\epsilon$,$T$  },
			ytick={-1,0,1},
			yticklabels={$-1$,$0$,$1$,},
			ylabel style={yshift=-0.3cm},
			legend style={at={(0.55,0.97)},anchor=north west,nodes=right},
			grid=none]
		\addplot[thick,color=u1Color,opacity=0.9,domain=0:6.2832,samples=100] { sin(1*deg(x)) }; 
		\addlegendentry{$u_1$};
		\addplot[thick,color=green!70!black,opacity=0.9,domain=0:6.2832,samples=100] { cos(1*deg(x)) };
		\addlegendentry{$u_2$};
		\addplot[thick,color=blue,opacity=0.9,const plot mark right,domain=0:6.2832,samples=21] { cos(1*deg(x)) };
		\addlegendentry{$\bar{u}_2$};
		
		\addplot[ultra thin,color=blue,opacity=0,const plot,fill=blue,fill opacity=0.2] coordinates {(0,0.9511) (0.3142,0.9511)} \closedcycle;
		\addplot[ultra thin,color=blue,opacity=0,const plot,fill=blue,fill opacity=0.2] coordinates {(\T+0,-0.9511) (\T+0.3142,-0.9511)} \closedcycle;
		\addplot[ultra thin,color=red,opacity=0,const plot,fill=red,fill opacity=0.2] coordinates {(0.3142,0.8090) (0.6283,0.8090)} \closedcycle;
		\addplot[ultra thin,color=red,opacity=0,const plot,fill=red,fill opacity=0.2] coordinates {(\T+0.3142,-0.8090) (\T+0.6283,-0.8090)} \closedcycle;
		\addplot[ultra thin,color=lime,opacity=0,const plot,fill=lime,fill opacity=0.2] coordinates {(0.6283,0.5878) (0.9425,0.5878)} \closedcycle;
		\addplot[ultra thin,color=lime,opacity=0,const plot,fill=lime,fill opacity=0.2] coordinates {(\T+0.6283,-0.5878) (\T+0.9425,-0.5878)} \closedcycle;
		\addplot[ultra thin,color=black,opacity=0,const plot,fill=black,fill opacity=0.2] coordinates {(0.9425,0.3090) (1.2566,0.3090)} \closedcycle;
		\addplot[ultra thin,color=black,opacity=0,const plot,fill=black,fill opacity=0.2] coordinates {(\T+0.9425,-0.3090) (\T+1.2566,-0.3090)} \closedcycle;
		\addplot[ultra thin,color=cyan,opacity=0,const plot,fill=cyan,fill opacity=0.2] coordinates {(1.2566,0) (1.5708,0)} \closedcycle;
		\addplot[ultra thin,color=cyan,opacity=0,const plot,fill=cyan,fill opacity=0.2] coordinates {(\T+1.2566,0) (\T+1.5708,0)} \closedcycle;
		\addplot[ultra thin,color=magenta,opacity=0,const plot,fill=magenta,fill opacity=0.2] coordinates {(\T+1.5708,0.3090) (\T+1.8850,0.3090)} \closedcycle;
		\addplot[ultra thin,color=magenta,opacity=0,const plot,fill=magenta,fill opacity=0.2] coordinates {(1.5708,-0.3090) (1.8850,-0.3090)} \closedcycle;
		\addplot[ultra thin,color=brown,opacity=0,const plot,fill=brown,fill opacity=0.2] coordinates {(1.8850,-0.5878) (2.1991,-0.5878)} \closedcycle;
		\addplot[ultra thin,color=brown,opacity=0,const plot,fill=brown,fill opacity=0.2] coordinates {(\T+1.8850,0.5878) (\T+2.1991,0.5878)} \closedcycle;
		\addplot[ultra thin,color=teal,opacity=0,const plot,fill=teal,fill opacity=0.2] coordinates {(2.1991,-0.8090) (2.5133,-0.8090)} \closedcycle;
		\addplot[ultra thin,color=teal,opacity=0,const plot,fill=teal,fill opacity=0.2] coordinates {(\T+2.1991,0.8090) (\T+2.5133,0.8090)} \closedcycle;
		\addplot[ultra thin,color=violet,opacity=0,const plot,fill=violet,fill opacity=0.2] coordinates {(2.5133,-0.9511) (2.8274,-0.9511)} \closedcycle;
		\addplot[ultra thin,color=violet,opacity=0,const plot,fill=violet,fill opacity=0.2] coordinates {(\T+2.5133,0.9511) (\T+2.8274,0.9511)} \closedcycle;
		\addplot[ultra thin,color=green,opacity=0,const plot,fill=green,fill opacity=0.2] coordinates {(2.8274,-1) (3.1416,-1)} \closedcycle;
		\addplot[ultra thin,color=green,opacity=0,const plot,fill=green,fill opacity=0.2] coordinates {(\T+2.8274,1) (\T+3.1416,1)} \closedcycle;
			
		\draw[-] (axis cs: -1,0) -- (axis cs: 6.28,0);
		\draw[-] (axis cs: 0,-1.5) -- (axis cs: 0,1.5);
	\end{axis} 
\end{tikzpicture}

%% file: Abbildungen/u1lu2lsequential.tex
	\colorlet{1Color}{orange}
	\colorlet{2Color}{green!70!black}
	\colorlet{3Color}{blue!70!white}
	\def\pii{3.1416}
	\begin{tikzpicture}[>=latex]
	\begin{axis}[
	width=\columnwidth,
	height=6cm,
	xmin=0, 
	xmax=18.8496,
	ymin=-1.3, 
	ymax=1.3,
	disabledatascaling,
	xlabel=Time $t$,
	ylabel=$u_{\cdot i}(t)$,
	xtick={0,3.1416,6.2832,9.4248,12.5664,15.708,18.8496},
	xticklabels={ $0$,$\tfrac{T}{2}$,$T$,$\tfrac{3T}{2}$,$2T$,$\tfrac{5T}{2}$,$3T$  },
	ytick={-3,-1,0,1,3},
	yticklabels={$-\alpha$,$-1$,$0$,$1$,$\alpha$},
	ylabel style={yshift=-0.3cm},
	legend style={at={(1,1.03)}, anchor=south east},
	legend columns = 3,
	grid=none,]
	
	\addplot[thick,color=1Color,opacity=0.9,domain=0:6.2832,samples=100] { sin(1*deg(x)) };
	\addlegendentry{$i=1$}
	\addplot[thick,color=2Color,opacity=0.9,domain=6.2832:12.5664,samples=100] { sin(1*deg(x)) }; 
	\addlegendentry{$i=2$}
	\addplot[thick,color=3Color,opacity=0.9,domain=12.5664:18.8496,samples=100] { sin(1*deg(x)) }; 
	\addlegendentry{$i=3$}
	
	\addplot[thick,color=1Color,opacity=0.9] coordinates {(6.2832,0.012) (18.8496,0.012)};
	\addplot[thick,color=2Color,opacity=0.9] coordinates {(0,0.012) (6.2832,0.012)};
	\addplot[thick,color=2Color,opacity=0.9] coordinates {(12.5664,-0.012) (18.8496,-0.012)};
	\addplot[thick,color=3Color,opacity=0.9] coordinates {(0,-0.012) (12.5664,-0.012)};
	
	\addplot[thick,color=1Color,opacity=0.9,domain=0:6.2832,samples=100,dashed] { cos(1*deg(x)) }; 
	\addplot[thick,color=2Color,opacity=0.9,domain=6.2832:12.5664,samples=100,dashed] { cos(1*deg(x)) }; 
	\addplot[thick,color=3Color,opacity=0.9,domain=12.5664:18.8496,samples=100,dashed] { cos(1*deg(x)) }; 
	
	\draw[-] (axis cs: -1,0) -- (axis cs: 19,0);
	\draw[-] (axis cs: 0,-1.5) -- (axis cs: 0,1.3);
	\end{axis}
	\end{tikzpicture}

%% file: Abbildungen/F.tex
%
%
\begin{tikzpicture}

\begin{axis}[%
width=\columnwidth,
height=3.5cm,
at={(0\columnwidth,0cm)},
xlabel={$x$},
ylabel={$F_2(x)$},
ylabel style={font=\color{white!15!black},at={(axis description cs:0.1,.5)},anchor=south},
xmin=-0.2,
xmax=2,
ymin=0,
ymax=2,
axis background/.style={fill=white}
]
\addplot [color=orange, forget plot]
  table[row sep=crcr]{%
-0.409	0.0836404999999999\\
-0.408	0.083232\\
-0.407	0.0828245\\
-0.406	0.0824180000000001\\
-0.405	0.0820124999999999\\
-0.404	0.081608\\
-0.403	0.0812045\\
-0.402	0.0808020000000001\\
-0.401	0.0804004999999999\\
-0.4	0.08\\
-0.399	0.0796005\\
-0.398	0.0792020000000001\\
-0.397	0.0788044999999999\\
-0.396	0.078408\\
-0.395	0.0780125\\
-0.394	0.077618\\
-0.393	0.0772244999999999\\
-0.392	0.076832\\
-0.391	0.0764405\\
-0.39	0.07605\\
-0.389	0.0756604999999999\\
-0.388	0.075272\\
-0.387	0.0748845\\
-0.386	0.0744980000000001\\
-0.385	0.0741124999999999\\
-0.384	0.073728\\
-0.383	0.0733445\\
-0.382	0.072962\\
-0.381	0.0725804999999999\\
-0.38	0.0722\\
-0.379	0.0718205\\
-0.378	0.071442\\
-0.377	0.0710644999999999\\
-0.376	0.070688\\
-0.375	0.0703125\\
-0.374	0.069938\\
-0.373	0.0695644999999999\\
-0.372	0.069192\\
-0.371	0.0688205\\
-0.37	0.06845\\
-0.369	0.0680804999999999\\
-0.368	0.067712\\
-0.367	0.0673445\\
-0.366	0.066978\\
-0.365	0.0666124999999999\\
-0.364	0.066248\\
-0.363	0.0658845\\
-0.362	0.065522\\
-0.361	0.0651604999999999\\
-0.36	0.0648\\
-0.359	0.0644405\\
-0.358	0.064082\\
-0.357	0.0637244999999999\\
-0.356	0.063368\\
-0.355	0.0630125\\
-0.354	0.062658\\
-0.353	0.0623044999999999\\
-0.352	0.061952\\
-0.351	0.0616005\\
-0.35	0.06125\\
-0.349	0.0609004999999999\\
-0.348	0.060552\\
-0.347	0.0602045\\
-0.346	0.059858\\
-0.345	0.0595124999999999\\
-0.344	0.0591679999999999\\
-0.343	0.0588245\\
-0.342	0.058482\\
-0.341	0.0581404999999999\\
-0.34	0.0577999999999999\\
-0.339	0.0574605\\
-0.338	0.057122\\
-0.337	0.0567844999999999\\
-0.336	0.0564479999999999\\
-0.335	0.0561125\\
-0.334	0.055778\\
-0.333	0.0554444999999999\\
-0.332	0.0551119999999999\\
-0.331	0.0547805\\
-0.33	0.05445\\
-0.329	0.0541204999999999\\
-0.328	0.053792\\
-0.327	0.0534645\\
-0.326	0.053138\\
-0.325	0.0528124999999999\\
-0.324	0.052488\\
-0.323	0.0521645\\
-0.322	0.051842\\
-0.321	0.0515204999999999\\
-0.32	0.0512\\
-0.319	0.0508805\\
-0.318	0.050562\\
-0.317	0.0502444999999999\\
-0.316	0.049928\\
-0.315	0.0496125\\
-0.314	0.049298\\
-0.313	0.0489844999999999\\
-0.312	0.0486719999999999\\
-0.311	0.0483605\\
-0.31	0.04805\\
-0.309	0.0477404999999999\\
-0.308	0.0474319999999999\\
-0.307	0.0471245\\
-0.306	0.046818\\
-0.305	0.0465124999999999\\
-0.304	0.0462079999999999\\
-0.303	0.0459045\\
-0.302	0.045602\\
-0.301	0.0453004999999999\\
-0.3	0.0449999999999999\\
-0.299	0.0447005\\
-0.298	0.044402\\
-0.297	0.0441044999999999\\
-0.296	0.043808\\
-0.295	0.0435125\\
-0.294	0.043218\\
-0.293	0.0429244999999999\\
-0.292	0.0426319999999999\\
-0.291	0.0423405\\
-0.29	0.04205\\
-0.289	0.0417604999999999\\
-0.288	0.0414719999999999\\
-0.287	0.0411845\\
-0.286	0.040898\\
-0.285	0.0406125\\
-0.284	0.0403279999999999\\
-0.283	0.0400445\\
-0.282	0.039762\\
-0.281	0.0394805\\
-0.28	0.0391999999999999\\
-0.279	0.0389205\\
-0.278	0.038642\\
-0.277	0.0383645\\
-0.276	0.0380879999999999\\
-0.275	0.0378125\\
-0.274	0.037538\\
-0.273	0.0372645\\
-0.272	0.0369919999999999\\
-0.271	0.0367205\\
-0.27	0.03645\\
-0.269	0.0361805\\
-0.268	0.0359119999999999\\
-0.267	0.0356445\\
-0.266	0.035378\\
-0.265	0.0351125\\
-0.264	0.0348479999999999\\
-0.263	0.0345845\\
-0.262	0.034322\\
-0.261	0.0340605\\
-0.26	0.0337999999999999\\
-0.259	0.0335405\\
-0.258	0.033282\\
-0.257	0.0330245\\
-0.256	0.0327679999999999\\
-0.255	0.0325125\\
-0.254	0.032258\\
-0.253	0.0320045\\
-0.252	0.0317519999999999\\
-0.251	0.0315005\\
-0.25	0.03125\\
-0.249	0.0310005\\
-0.248	0.0307519999999999\\
-0.247	0.0305045\\
-0.246	0.030258\\
-0.245	0.0300125\\
-0.244	0.0297679999999999\\
-0.243	0.0295245\\
-0.242	0.029282\\
-0.241	0.0290405\\
-0.24	0.0287999999999999\\
-0.239	0.0285605\\
-0.238	0.028322\\
-0.237	0.0280845\\
-0.236	0.0278479999999999\\
-0.235	0.0276125\\
-0.234	0.027378\\
-0.233	0.0271445\\
-0.232	0.0269119999999999\\
-0.231	0.0266805\\
-0.23	0.02645\\
-0.229	0.0262205\\
-0.228	0.0259919999999999\\
-0.227	0.0257645\\
-0.226	0.025538\\
-0.225	0.0253125\\
-0.224	0.0250879999999999\\
-0.223	0.0248645\\
-0.222	0.024642\\
-0.221	0.0244205\\
-0.22	0.0241999999999999\\
-0.219	0.0239805\\
-0.218	0.023762\\
-0.217	0.0235445\\
-0.216	0.0233279999999999\\
-0.215	0.0231125\\
-0.214	0.022898\\
-0.213	0.0226845\\
-0.212	0.0224719999999999\\
-0.211	0.0222605\\
-0.21	0.02205\\
-0.209	0.0218405\\
-0.208	0.0216319999999999\\
-0.207	0.0214245\\
-0.206	0.021218\\
-0.205	0.0210125\\
-0.204	0.0208079999999999\\
-0.203	0.0206045\\
-0.202	0.020402\\
-0.201	0.0202005\\
-0.2	0.0199999999999999\\
-0.199	0.0198005\\
-0.198	0.019602\\
-0.197	0.0194045\\
-0.196	0.0192079999999999\\
-0.195	0.0190125\\
-0.194	0.018818\\
-0.193	0.0186245\\
-0.192	0.0184319999999999\\
-0.191	0.0182405\\
-0.19	0.01805\\
-0.189	0.0178605\\
-0.188	0.0176719999999999\\
-0.187	0.0174845\\
-0.186	0.017298\\
-0.185	0.0171125\\
-0.184	0.0169279999999999\\
-0.183	0.0167445\\
-0.182	0.016562\\
-0.181	0.0163805\\
-0.18	0.0162\\
-0.179	0.0160205\\
-0.178	0.015842\\
-0.177	0.0156645\\
-0.176	0.0154879999999999\\
-0.175	0.0153125\\
-0.174	0.015138\\
-0.173	0.0149645\\
-0.172	0.0147919999999999\\
-0.171	0.0146205\\
-0.17	0.01445\\
-0.169	0.0142805\\
-0.168	0.014112\\
-0.167	0.0139445\\
-0.166	0.013778\\
-0.165	0.0136125\\
-0.164	0.013448\\
-0.163	0.0132845\\
-0.162	0.013122\\
-0.161	0.0129605\\
-0.16	0.0128\\
-0.159	0.0126405\\
-0.158	0.012482\\
-0.157	0.0123245\\
-0.156	0.012168\\
-0.155	0.0120125\\
-0.154	0.011858\\
-0.153	0.0117045\\
-0.152	0.011552\\
-0.151	0.0114005\\
-0.15	0.01125\\
-0.149	0.0111005\\
-0.148	0.010952\\
-0.147	0.0108045\\
-0.146	0.010658\\
-0.145	0.0105125\\
-0.144	0.010368\\
-0.143	0.0102245\\
-0.142	0.010082\\
-0.141	0.0099405\\
-0.14	0.00980000000000002\\
-0.139	0.00966049999999997\\
-0.138	0.00952199999999999\\
-0.137	0.0093845\\
-0.136	0.00924800000000002\\
-0.135	0.00911249999999997\\
-0.134	0.00897799999999999\\
-0.133	0.0088445\\
-0.132	0.00871200000000001\\
-0.131	0.00858049999999997\\
-0.13	0.00844999999999999\\
-0.129	0.0083205\\
-0.128	0.00819200000000002\\
-0.127	0.00806449999999997\\
-0.126	0.00793799999999999\\
-0.125	0.0078125\\
-0.124	0.00768800000000001\\
-0.123	0.00756449999999997\\
-0.122	0.00744199999999999\\
-0.121	0.0073205\\
-0.12	0.00720000000000001\\
-0.119	0.00708049999999997\\
-0.118	0.00696199999999999\\
-0.117	0.0068445\\
-0.116	0.00672800000000001\\
-0.115	0.00661249999999997\\
-0.114	0.00649799999999999\\
-0.113	0.0063845\\
-0.112	0.00627200000000001\\
-0.111	0.00616049999999997\\
-0.11	0.00604999999999999\\
-0.109	0.0059405\\
-0.108	0.00583200000000001\\
-0.107	0.00572449999999998\\
-0.106	0.00561799999999999\\
-0.105	0.0055125\\
-0.104	0.00540800000000001\\
-0.103	0.00530449999999998\\
-0.102	0.00520199999999999\\
-0.101	0.0051005\\
-0.1	0.00500000000000001\\
-0.0989999999999998	0.00490049999999998\\
-0.0979999999999999	0.00480199999999999\\
-0.097	0.0047045\\
-0.0960000000000001	0.00460800000000001\\
-0.0949999999999998	0.00451249999999998\\
-0.0939999999999999	0.00441799999999999\\
-0.093	0.0043245\\
-0.0920000000000001	0.00423200000000001\\
-0.0909999999999997	0.00414049999999998\\
-0.0899999999999999	0.00404999999999999\\
-0.089	0.0039605\\
-0.0880000000000001	0.00387200000000001\\
-0.0869999999999997	0.00378449999999998\\
-0.0859999999999999	0.00369799999999999\\
-0.085	0.0036125\\
-0.0840000000000001	0.00352800000000001\\
-0.0829999999999997	0.00344449999999998\\
-0.0819999999999999	0.00336199999999999\\
-0.081	0.0032805\\
-0.0800000000000001	0.00320000000000001\\
-0.0789999999999997	0.00312049999999998\\
-0.0779999999999998	0.00304199999999999\\
-0.077	0.0029645\\
-0.0760000000000001	0.00288800000000001\\
-0.0749999999999997	0.00281249999999998\\
-0.0739999999999998	0.00273799999999999\\
-0.073	0.0026645\\
-0.0720000000000001	0.002592\\
-0.0709999999999997	0.00252049999999998\\
-0.0699999999999998	0.00244999999999999\\
-0.069	0.0023805\\
-0.0680000000000001	0.002312\\
-0.0669999999999997	0.00224449999999998\\
-0.0659999999999998	0.00217799999999999\\
-0.0649999999999999	0.0021125\\
-0.0640000000000001	0.002048\\
-0.0629999999999997	0.00198449999999998\\
-0.0619999999999998	0.00192199999999999\\
-0.0609999999999999	0.0018605\\
-0.0600000000000001	0.0018\\
-0.0589999999999997	0.00174049999999998\\
-0.0579999999999998	0.00168199999999999\\
-0.0569999999999999	0.0016245\\
-0.056	0.001568\\
-0.0549999999999997	0.00151249999999998\\
-0.0539999999999998	0.00145799999999999\\
-0.0529999999999999	0.0014045\\
-0.052	0.001352\\
-0.0509999999999997	0.00130049999999999\\
-0.0499999999999998	0.00124999999999999\\
-0.0489999999999999	0.0012005\\
-0.048	0.001152\\
-0.0469999999999997	0.00110449999999999\\
-0.0459999999999998	0.00105799999999999\\
-0.0449999999999999	0.0010125\\
-0.044	0.000968000000000002\\
-0.0429999999999997	0.000924499999999987\\
-0.0419999999999998	0.000881999999999992\\
-0.0409999999999999	0.000840499999999997\\
-0.04	0.000800000000000001\\
-0.0389999999999997	0.000760499999999988\\
-0.0379999999999998	0.000721999999999993\\
-0.0369999999999999	0.000684499999999997\\
-0.036	0.000648000000000001\\
-0.0349999999999997	0.000612499999999989\\
-0.0339999999999998	0.000577999999999993\\
-0.0329999999999999	0.000544499999999997\\
-0.032	0.000512000000000001\\
-0.0310000000000001	0.000480500000000004\\
-0.0299999999999998	0.000449999999999994\\
-0.0289999999999999	0.000420499999999998\\
-0.028	0.000392000000000001\\
-0.0270000000000001	0.000364500000000004\\
-0.0259999999999998	0.000337999999999995\\
-0.0249999999999999	0.000312499999999998\\
-0.024	0.000288\\
-0.0230000000000001	0.000264500000000003\\
-0.0219999999999998	0.000241999999999996\\
-0.0209999999999999	0.000220499999999998\\
-0.02	0.0002\\
-0.0190000000000001	0.000180500000000002\\
-0.0179999999999998	0.000161999999999996\\
-0.0169999999999999	0.000144499999999998\\
-0.016	0.000128\\
-0.0150000000000001	0.000112500000000002\\
-0.0139999999999998	9.79999999999971e-05\\
-0.0129999999999999	8.44999999999987e-05\\
-0.012	7.20000000000001e-05\\
-0.0110000000000001	6.05000000000013e-05\\
-0.00999999999999979	4.99999999999979e-05\\
-0.0089999999999999	4.04999999999991e-05\\
-0.00800000000000001	3.20000000000001e-05\\
-0.00700000000000012	2.45000000000008e-05\\
-0.00599999999999978	1.79999999999987e-05\\
-0.00499999999999989	1.24999999999995e-05\\
-0.004	8.00000000000001e-06\\
-0.00300000000000011	4.50000000000034e-06\\
-0.00199999999999978	1.99999999999956e-06\\
-0.00099999999999989	4.9999999999989e-07\\
0	2.2372502603127e-36\\
};
\addplot [color=orange, forget plot]
  table[row sep=crcr]{%
0	2.2372502603127e-36\\
0.00099999999999989	4.9999999999989e-07\\
0.00199999999999978	1.99999999999956e-06\\
0.00300000000000011	4.50000000000034e-06\\
0.004	8.00000000000001e-06\\
0.00499999999999989	1.24999999999995e-05\\
0.00599999999999978	1.79999999999987e-05\\
0.00700000000000012	2.45000000000008e-05\\
0.00800000000000001	3.20000000000001e-05\\
0.0089999999999999	4.04999999999991e-05\\
0.00999999999999979	4.99999999999979e-05\\
0.0110000000000001	6.05000000000013e-05\\
0.012	7.20000000000001e-05\\
0.0129999999999999	8.44999999999987e-05\\
0.0139999999999998	9.79999999999971e-05\\
0.0150000000000001	0.000112500000000002\\
0.016	0.000128\\
0.0169999999999999	0.000144499999999998\\
0.0179999999999998	0.000161999999999996\\
0.0190000000000001	0.000180500000000002\\
0.02	0.0002\\
0.0209999999999999	0.000220499999999998\\
0.0219999999999998	0.000241999999999996\\
0.0230000000000001	0.000264500000000003\\
0.024	0.000288\\
0.0249999999999999	0.000312499999999998\\
0.0259999999999998	0.000337999999999995\\
0.0270000000000001	0.000364500000000004\\
0.028	0.000392000000000001\\
0.0289999999999999	0.000420499999999998\\
0.0299999999999998	0.000449999999999994\\
0.0310000000000001	0.000480500000000004\\
0.032	0.000512000000000001\\
0.0329999999999999	0.000544499999999997\\
0.0339999999999998	0.000577999999999993\\
0.0349999999999997	0.000612499999999989\\
0.036	0.000648000000000001\\
0.0369999999999999	0.000684499999999997\\
0.0379999999999998	0.000721999999999993\\
0.0389999999999997	0.000760499999999988\\
0.04	0.000800000000000001\\
0.0409999999999999	0.000840499999999997\\
0.0419999999999998	0.000881999999999992\\
0.0429999999999997	0.000924499999999987\\
0.044	0.000968000000000002\\
0.0449999999999999	0.0010125\\
0.0459999999999998	0.00105799999999999\\
0.0469999999999997	0.00110449999999999\\
0.048	0.001152\\
0.0489999999999999	0.0012005\\
0.0499999999999998	0.00124999999999999\\
0.0509999999999997	0.00130049999999999\\
0.052	0.001352\\
0.0529999999999999	0.0014045\\
0.0539999999999998	0.00145799999999999\\
0.0549999999999997	0.00151249999999998\\
0.056	0.001568\\
0.0569999999999999	0.0016245\\
0.0579999999999998	0.00168199999999999\\
0.0589999999999997	0.00174049999999998\\
0.0600000000000001	0.0018\\
0.0609999999999999	0.0018605\\
0.0619999999999998	0.00192199999999999\\
0.0629999999999997	0.00198449999999998\\
0.0640000000000001	0.002048\\
0.0649999999999999	0.0021125\\
0.0659999999999998	0.00217799999999999\\
0.0669999999999997	0.00224449999999998\\
0.0680000000000001	0.002312\\
0.069	0.0023805\\
0.0699999999999998	0.00244999999999999\\
0.0709999999999997	0.00252049999999998\\
0.0720000000000001	0.002592\\
0.073	0.0026645\\
0.0739999999999998	0.00273799999999999\\
0.0749999999999997	0.00281249999999998\\
0.0760000000000001	0.00288800000000001\\
0.077	0.0029645\\
0.0779999999999998	0.00304199999999999\\
0.0789999999999997	0.00312049999999998\\
0.0800000000000001	0.00320000000000001\\
0.081	0.0032805\\
0.0819999999999999	0.00336199999999999\\
0.0829999999999997	0.00344449999999998\\
0.0840000000000001	0.00352800000000001\\
0.085	0.0036125\\
0.0859999999999999	0.00369799999999999\\
0.0869999999999997	0.00378449999999998\\
0.0880000000000001	0.00387200000000001\\
0.089	0.0039605\\
0.0899999999999999	0.00404999999999999\\
0.0909999999999997	0.00414049999999998\\
0.0920000000000001	0.00423200000000001\\
0.093	0.0043245\\
0.0939999999999999	0.00441799999999999\\
0.0949999999999998	0.00451249999999998\\
0.0960000000000001	0.00460800000000001\\
0.097	0.0047045\\
0.0979999999999999	0.00480199999999999\\
0.0989999999999998	0.00490049999999998\\
0.1	0.00500000000000001\\
0.101	0.0051005\\
0.102	0.00520199999999999\\
0.103	0.00530449999999998\\
0.104	0.00540800000000001\\
0.105	0.0055125\\
0.106	0.00561799999999999\\
0.107	0.00572449999999998\\
0.108	0.00583200000000001\\
0.109	0.0059405\\
0.11	0.00604999999999999\\
0.111	0.00616049999999997\\
0.112	0.00627200000000001\\
0.113	0.0063845\\
0.114	0.00649799999999999\\
0.115	0.00661249999999997\\
0.116	0.00672800000000001\\
0.117	0.0068445\\
0.118	0.00696199999999999\\
0.119	0.00708049999999997\\
0.12	0.00720000000000001\\
0.121	0.0073205\\
0.122	0.00744199999999999\\
0.123	0.00756449999999997\\
0.124	0.00768800000000001\\
0.125	0.0078125\\
0.126	0.00793799999999999\\
0.127	0.00806449999999997\\
0.128	0.00819200000000001\\
0.129	0.0083205\\
0.13	0.00844999999999999\\
0.131	0.00858049999999997\\
0.132	0.00871200000000001\\
0.133	0.0088445\\
0.134	0.00897799999999999\\
0.135	0.00911249999999997\\
0.136	0.00924800000000001\\
0.137	0.0093845\\
0.138	0.00952199999999999\\
0.139	0.00966049999999997\\
0.14	0.00980000000000002\\
0.141	0.0099405\\
0.142	0.010082\\
0.143	0.0102245\\
0.144	0.010368\\
0.145	0.0105125\\
0.146	0.010658\\
0.147	0.0108045\\
0.148	0.010952\\
0.149	0.0111005\\
0.15	0.01125\\
0.151	0.0114005\\
0.152	0.011552\\
0.153	0.0117045\\
0.154	0.011858\\
0.155	0.0120125\\
0.156	0.012168\\
0.157	0.0123245\\
0.158	0.012482\\
0.159	0.0126405\\
0.16	0.0128\\
0.161	0.0129605\\
0.162	0.013122\\
0.163	0.0132845\\
0.164	0.013448\\
0.165	0.0136125\\
0.166	0.013778\\
0.167	0.0139445\\
0.168	0.014112\\
0.169	0.0142805\\
0.17	0.01445\\
0.171	0.0146205\\
0.172	0.0147919999999999\\
0.173	0.0149645\\
0.174	0.015138\\
0.175	0.0153125\\
0.176	0.0154879999999999\\
0.177	0.0156645\\
0.178	0.015842\\
0.179	0.0160205\\
0.18	0.0161999999999999\\
0.181	0.0163805\\
0.182	0.016562\\
0.183	0.0167445\\
0.184	0.0169279999999999\\
0.185	0.0171125\\
0.186	0.017298\\
0.187	0.0174845\\
0.188	0.0176719999999999\\
0.189	0.0178605\\
0.19	0.01805\\
0.191	0.0182405\\
0.192	0.0184319999999999\\
0.193	0.0186245\\
0.194	0.018818\\
0.195	0.0190125\\
0.196	0.0192079999999999\\
0.197	0.0194045\\
0.198	0.019602\\
0.199	0.0198005\\
0.2	0.0199999999999999\\
0.201	0.0202005\\
0.202	0.020402\\
0.203	0.0206045\\
0.204	0.0208079999999999\\
0.205	0.0210125\\
0.206	0.021218\\
0.207	0.0214245\\
0.208	0.0216319999999999\\
0.209	0.0218405\\
0.21	0.02205\\
0.211	0.0222605\\
0.212	0.0224719999999999\\
0.213	0.0226845\\
0.214	0.022898\\
0.215	0.0231125\\
0.216	0.0233279999999999\\
0.217	0.0235445\\
0.218	0.023762\\
0.219	0.0239805\\
0.22	0.0241999999999999\\
0.221	0.0244205\\
0.222	0.024642\\
0.223	0.0248645\\
0.224	0.0250879999999999\\
0.225	0.0253125\\
0.226	0.025538\\
0.227	0.0257645\\
0.228	0.0259919999999999\\
0.229	0.0262205\\
0.23	0.02645\\
0.231	0.0266805\\
0.232	0.0269119999999999\\
0.233	0.0271445\\
0.234	0.027378\\
0.235	0.0276125\\
0.236	0.0278479999999999\\
0.237	0.0280845\\
0.238	0.028322\\
0.239	0.0285605\\
0.24	0.0287999999999999\\
0.241	0.0290405\\
0.242	0.029282\\
0.243	0.0295245\\
0.244	0.0297679999999999\\
0.245	0.0300125\\
0.246	0.030258\\
0.247	0.0305045\\
0.248	0.0307519999999999\\
0.249	0.0310005\\
0.25	0.03125\\
0.251	0.0315005\\
0.252	0.0317519999999999\\
0.253	0.0320045\\
0.254	0.032258\\
0.255	0.0325125\\
0.256	0.0327679999999999\\
0.257	0.0330245\\
0.258	0.033282\\
0.259	0.0335405\\
0.26	0.0337999999999999\\
0.261	0.0340605\\
0.262	0.034322\\
0.263	0.0345845\\
0.264	0.0348479999999999\\
0.265	0.0351125\\
0.266	0.035378\\
0.267	0.0356445\\
0.268	0.0359119999999999\\
0.269	0.0361805\\
0.27	0.03645\\
0.271	0.0367205\\
0.272	0.0369919999999999\\
0.273	0.0372645\\
0.274	0.037538\\
0.275	0.0378125\\
0.276	0.0380879999999999\\
0.277	0.0383645\\
0.278	0.038642\\
0.279	0.0389205\\
0.28	0.0392\\
0.281	0.0394805\\
0.282	0.039762\\
0.283	0.0400445\\
0.284	0.0403279999999999\\
0.285	0.0406125\\
0.286	0.040898\\
0.287	0.0411845\\
0.288	0.0414719999999999\\
0.289	0.0417604999999999\\
0.29	0.04205\\
0.291	0.0423405\\
0.292	0.0426319999999999\\
0.293	0.0429244999999999\\
0.294	0.043218\\
0.295	0.0435125\\
0.296	0.0438079999999999\\
0.297	0.0441044999999999\\
0.298	0.044402\\
0.299	0.0447005\\
0.3	0.0449999999999999\\
0.301	0.0453004999999999\\
0.302	0.045602\\
0.303	0.0459045\\
0.304	0.0462079999999999\\
0.305	0.0465124999999999\\
0.306	0.046818\\
0.307	0.0471245\\
0.308	0.0474319999999999\\
0.309	0.0477404999999999\\
0.31	0.04805\\
0.311	0.0483605\\
0.312	0.0486719999999999\\
0.313	0.0489844999999999\\
0.314	0.049298\\
0.315	0.0496125\\
0.316	0.049928\\
0.317	0.0502444999999999\\
0.318	0.050562\\
0.319	0.0508805\\
0.32	0.0511999999999999\\
0.321	0.0515204999999999\\
0.322	0.051842\\
0.323	0.0521645\\
0.324	0.052488\\
0.325	0.0528124999999999\\
0.326	0.053138\\
0.327	0.0534645\\
0.328	0.053792\\
0.329	0.0541204999999999\\
0.33	0.05445\\
0.331	0.0547805\\
0.332	0.055112\\
0.333	0.0554444999999999\\
0.334	0.055778\\
0.335	0.0561125\\
0.336	0.056448\\
0.337	0.0567844999999999\\
0.338	0.057122\\
0.339	0.0574605\\
0.34	0.0578\\
0.341	0.0581404999999999\\
0.342	0.058482\\
0.343	0.0588245\\
0.344	0.059168\\
0.345	0.0595124999999999\\
0.346	0.059858\\
0.347	0.0602045\\
0.348	0.060552\\
0.349	0.0609004999999999\\
0.35	0.06125\\
0.351	0.0616005\\
0.352	0.061952\\
0.353	0.0623044999999999\\
0.354	0.062658\\
0.355	0.0630125\\
0.356	0.063368\\
0.357	0.0637244999999999\\
0.358	0.064082\\
0.359	0.0644405\\
0.36	0.0648\\
0.361	0.0651604999999999\\
0.362	0.065522\\
0.363	0.0658845\\
0.364	0.066248\\
0.365	0.0666124999999999\\
0.366	0.066978\\
0.367	0.0673445\\
0.368	0.067712\\
0.369	0.0680804999999999\\
0.37	0.06845\\
0.371	0.0688205\\
0.372	0.069192\\
0.373	0.0695644999999999\\
0.374	0.069938\\
0.375	0.0703125\\
0.376	0.070688\\
0.377	0.0710644999999999\\
0.378	0.071442\\
0.379	0.0718205\\
0.38	0.0722\\
0.381	0.0725804999999999\\
0.382	0.072962\\
0.383	0.0733445\\
0.384	0.073728\\
0.385	0.0741124999999999\\
0.386	0.0744980000000001\\
0.387	0.0748845\\
0.388	0.075272\\
0.389	0.0756604999999999\\
0.39	0.07605\\
0.391	0.0764405\\
0.392	0.076832\\
0.393	0.0772244999999999\\
0.394	0.077618\\
0.395	0.0780125\\
0.396	0.078408\\
0.397	0.0788044999999999\\
0.398	0.0792020000000001\\
0.399	0.0796005\\
0.4	0.08\\
0.401	0.0804004999999999\\
0.402	0.0808020000000001\\
0.403	0.0812045\\
0.404	0.081608\\
0.405	0.0820124999999999\\
0.406	0.082418\\
0.407	0.0828245\\
0.408	0.083232\\
0.409	0.0836404999999999\\
0.41	0.0840500000000001\\
0.411	0.0844605\\
0.412	0.084872\\
0.413	0.0852844999999999\\
0.414	0.0856980000000001\\
0.415	0.0861125\\
0.416	0.086528\\
0.417	0.0869444999999999\\
0.418	0.0873619999999999\\
0.419	0.0877805\\
0.42	0.0882\\
0.421	0.0886204999999999\\
0.422	0.0890419999999999\\
0.423	0.0894645\\
0.424	0.089888\\
0.425	0.0903124999999999\\
0.426	0.0907379999999999\\
0.427	0.0911645\\
0.428	0.091592\\
0.429	0.0920204999999999\\
0.43	0.0924499999999999\\
0.431	0.0928805\\
0.432	0.093312\\
0.433	0.0937444999999999\\
0.434	0.0941779999999999\\
0.435	0.0946125\\
0.436	0.095048\\
0.437	0.0954844999999999\\
0.438	0.0959219999999999\\
0.439	0.0963605\\
0.44	0.0968\\
0.441	0.0972404999999999\\
0.442	0.0976819999999999\\
0.443	0.0981245\\
0.444	0.098568\\
0.445	0.0990124999999999\\
0.446	0.0994579999999999\\
0.447	0.0999045\\
0.448	0.100352\\
0.449	0.1008005\\
0.45	0.10125\\
0.451	0.1017005\\
0.452	0.102152\\
0.453	0.1026045\\
0.454	0.103058\\
0.455	0.1035125\\
0.456	0.103968\\
0.457	0.1044245\\
0.458	0.104882\\
0.459	0.1053405\\
0.46	0.1058\\
0.461	0.1062605\\
0.462	0.106722\\
0.463	0.1071845\\
0.464	0.107648\\
0.465	0.1081125\\
0.466	0.108578\\
0.467	0.1090445\\
0.468	0.109512\\
0.469	0.1099805\\
0.47	0.11045\\
0.471	0.1109205\\
0.472	0.111392\\
0.473	0.1118645\\
0.474	0.112338\\
0.475	0.1128125\\
0.476	0.113288\\
0.477	0.1137645\\
0.478	0.114242\\
0.479	0.1147205\\
0.48	0.1152\\
0.481	0.1156805\\
0.482	0.116162\\
0.483	0.1166445\\
0.484	0.117128\\
0.485	0.1176125\\
0.486	0.118098\\
0.487	0.1185845\\
0.488	0.119072\\
0.489	0.1195605\\
0.49	0.12005\\
0.491	0.1205405\\
0.492	0.121032\\
0.493	0.1215245\\
0.494	0.122018\\
0.495	0.1225125\\
0.496	0.123008\\
0.497	0.1235045\\
0.498	0.124002\\
0.499	0.1245005\\
0.5	0.125\\
0.501	0.1255005\\
0.502	0.126002\\
0.503	0.1265045\\
0.504	0.127008\\
0.505	0.1275125\\
0.506	0.128018\\
0.507	0.1285245\\
0.508	0.129032\\
0.509	0.1295405\\
0.51	0.13005\\
0.511	0.1305605\\
0.512	0.131072\\
0.513	0.1315845\\
0.514	0.132098\\
0.515	0.1326125\\
0.516	0.133128\\
0.517	0.1336445\\
0.518	0.134162\\
0.519	0.1346805\\
0.52	0.1352\\
0.521	0.1357205\\
0.522	0.136242\\
0.523	0.1367645\\
0.524	0.137288\\
0.525	0.1378125\\
0.526	0.138338\\
0.527	0.1388645\\
0.528	0.139392\\
0.529	0.1399205\\
0.53	0.14045\\
0.531	0.1409805\\
0.532	0.141512\\
0.533	0.1420445\\
0.534	0.142578\\
0.535	0.1431125\\
0.536	0.143648\\
0.537	0.1441845\\
0.538	0.144722\\
0.539	0.1452605\\
0.54	0.1458\\
0.541	0.1463405\\
0.542	0.146882\\
0.543	0.1474245\\
0.544	0.147968\\
0.545	0.1485125\\
0.546	0.149058\\
0.547	0.1496045\\
0.548	0.150152\\
0.549	0.1507005\\
0.55	0.15125\\
0.551	0.1518005\\
0.552	0.152352\\
0.553	0.1529045\\
0.554	0.153458\\
0.555	0.1540125\\
0.556	0.154568\\
0.557	0.1551245\\
0.558	0.155682\\
0.559	0.1562405\\
0.56	0.1568\\
0.561	0.1573605\\
0.562	0.157922\\
0.563	0.1584845\\
0.564	0.159048\\
0.565	0.1596125\\
0.566	0.160178\\
0.567	0.1607445\\
0.568	0.161312\\
0.569	0.1618805\\
0.57	0.16245\\
0.571	0.1630205\\
0.572	0.163592\\
0.573	0.1641645\\
0.574	0.164738\\
0.575	0.1653125\\
0.576	0.165888\\
0.577	0.1664645\\
0.578	0.167042\\
0.579	0.1676205\\
0.58	0.1682\\
0.581	0.1687805\\
0.582	0.169362\\
0.583	0.1699445\\
0.584	0.170528\\
0.585	0.1711125\\
0.586	0.171698\\
0.587	0.1722845\\
0.588	0.172872\\
0.589	0.1734605\\
0.59	0.17405\\
0.591	0.1746405\\
0.592	0.175232\\
0.593	0.1758245\\
0.594	0.176418\\
0.595	0.1770125\\
0.596	0.177608\\
0.597	0.1782045\\
0.598	0.178802\\
0.599	0.1794005\\
0.6	0.18\\
0.601	0.1806005\\
0.602	0.181202\\
0.603	0.1818045\\
0.604	0.182408\\
0.605	0.1830125\\
0.606	0.183618\\
0.607	0.1842245\\
0.608	0.184832\\
0.609	0.1854405\\
0.61	0.18605\\
0.611	0.1866605\\
0.612	0.187272\\
0.613	0.1878845\\
0.614	0.188498\\
0.615	0.1891125\\
0.616	0.189728\\
0.617	0.1903445\\
0.618	0.190962\\
0.619	0.1915805\\
0.62	0.1922\\
0.621	0.1928205\\
0.622	0.193442\\
0.623	0.1940645\\
0.624	0.194688\\
0.625	0.1953125\\
0.626	0.195938\\
0.627	0.1965645\\
0.628	0.197192\\
0.629	0.1978205\\
0.63	0.19845\\
0.631	0.1990805\\
0.632	0.199712\\
0.633	0.2003445\\
0.634	0.200978\\
0.635	0.2016125\\
0.636	0.202248\\
0.637	0.2028845\\
0.638	0.203522\\
0.639	0.2041605\\
0.64	0.2048\\
0.641	0.2054405\\
0.642	0.206082\\
0.643	0.2067245\\
0.644	0.207368\\
0.645	0.2080125\\
0.646	0.208658\\
0.647	0.2093045\\
0.648	0.209952\\
0.649	0.2106005\\
0.65	0.21125\\
0.651	0.2119005\\
0.652	0.212552\\
0.653	0.2132045\\
0.654	0.213858\\
0.655	0.2145125\\
0.656	0.215168\\
0.657	0.2158245\\
0.658	0.216482\\
0.659	0.2171405\\
0.66	0.2178\\
0.661	0.2184605\\
0.662	0.219122\\
0.663	0.2197845\\
0.664	0.220448\\
0.665	0.2211125\\
0.666	0.221778\\
0.667	0.2224445\\
0.668	0.223112\\
0.669	0.2237805\\
0.67	0.22445\\
0.671	0.2251205\\
0.672	0.225792\\
0.673	0.2264645\\
0.674	0.227138\\
0.675	0.2278125\\
0.676	0.228488\\
0.677	0.2291645\\
0.678	0.229842\\
0.679	0.2305205\\
0.68	0.2312\\
0.681	0.2318805\\
0.682	0.232562\\
0.683	0.2332445\\
0.684	0.233928\\
0.685	0.2346125\\
0.686	0.235298\\
0.687	0.2359845\\
0.688	0.236672\\
0.689	0.2373605\\
0.69	0.23805\\
0.691	0.2387405\\
0.692	0.239432\\
0.693	0.2401245\\
0.694	0.240818\\
0.695	0.2415125\\
0.696	0.242208\\
0.697	0.2429045\\
0.698	0.243602\\
0.699	0.2443005\\
0.7	0.245\\
0.701	0.2457005\\
0.702	0.246402\\
0.703	0.2471045\\
0.704	0.247808\\
0.705	0.2485125\\
0.706	0.249218\\
0.707	0.2499245\\
0.708	0.250632\\
0.709	0.2513405\\
0.71	0.25205\\
0.711	0.2527605\\
0.712	0.253472\\
0.713	0.2541845\\
0.714	0.254898\\
0.715	0.2556125\\
0.716	0.256328\\
0.717	0.2570445\\
0.718	0.257762\\
0.719	0.2584805\\
0.72	0.2592\\
0.721	0.2599205\\
0.722	0.260642\\
0.723	0.2613645\\
0.724	0.262088\\
0.725	0.2628125\\
0.726	0.263538\\
0.727	0.2642645\\
0.728	0.264992\\
0.729	0.2657205\\
0.73	0.26645\\
0.731	0.2671805\\
0.732	0.267912\\
0.733	0.2686445\\
0.734	0.269378\\
0.735	0.2701125\\
0.736	0.270848\\
0.737	0.2715845\\
0.738	0.272322\\
0.739	0.2730605\\
0.74	0.2738\\
0.741	0.2745405\\
0.742	0.275282\\
0.743	0.2760245\\
0.744	0.276768\\
0.745	0.2775125\\
0.746	0.278258\\
0.747	0.2790045\\
0.748	0.279752\\
0.749	0.2805005\\
0.75	0.28125\\
0.751	0.2820005\\
0.752	0.282752\\
0.753	0.2835045\\
0.754	0.284258\\
0.755	0.2850125\\
0.756	0.285768\\
0.757	0.2865245\\
0.758	0.287282\\
0.759	0.2880405\\
0.76	0.2888\\
0.761	0.2895605\\
0.762	0.290322\\
0.763	0.2910845\\
0.764	0.291848\\
0.765	0.2926125\\
0.766	0.293378\\
0.767	0.2941445\\
0.768	0.294912\\
0.769	0.2956805\\
0.77	0.29645\\
0.771	0.2972205\\
0.772	0.297992\\
0.773	0.2987645\\
0.774	0.299538\\
0.775	0.3003125\\
0.776	0.301088\\
0.777	0.3018645\\
0.778	0.302642\\
0.779	0.3034205\\
0.78	0.3042\\
0.781	0.3049805\\
0.782	0.305762\\
0.783	0.3065445\\
0.784	0.307328\\
0.785	0.3081125\\
0.786	0.308898\\
0.787	0.3096845\\
0.788	0.310472\\
0.789	0.3112605\\
0.79	0.31205\\
0.791	0.3128405\\
0.792	0.313632\\
0.793	0.3144245\\
0.794	0.315218\\
0.795	0.3160125\\
0.796	0.316808\\
0.797	0.3176045\\
0.798	0.318402\\
0.799	0.3192005\\
0.8	0.32\\
0.801	0.3208005\\
0.802	0.321602\\
0.803	0.3224045\\
0.804	0.323208\\
0.805	0.3240125\\
0.806	0.324818\\
0.807	0.3256245\\
0.808	0.326432\\
0.809	0.3272405\\
0.81	0.32805\\
0.811	0.3288605\\
0.812	0.329672\\
0.813	0.3304845\\
0.814	0.331298\\
0.815	0.3321125\\
0.816	0.332928\\
0.817	0.3337445\\
0.818	0.334562\\
0.819	0.3353805\\
0.82	0.3362\\
0.821	0.3370205\\
0.822	0.337842\\
0.823	0.3386645\\
0.824	0.339488\\
0.825	0.3403125\\
0.826	0.341138\\
0.827	0.3419645\\
0.828	0.342792\\
0.829	0.3436205\\
0.83	0.34445\\
0.831	0.3452805\\
0.832	0.346112\\
0.833	0.3469445\\
0.834	0.347778\\
0.835	0.3486125\\
0.836	0.349448\\
0.837	0.3502845\\
0.838	0.351122\\
0.839	0.3519605\\
0.84	0.3528\\
0.841	0.3536405\\
0.842	0.354482\\
0.843	0.3553245\\
0.844	0.356168\\
0.845	0.3570125\\
0.846	0.357858\\
0.847	0.3587045\\
0.848	0.359552\\
0.849	0.3604005\\
0.85	0.36125\\
0.851	0.3621005\\
0.852	0.362952\\
0.853	0.3638045\\
0.854	0.364658\\
0.855	0.3655125\\
0.856	0.366368\\
0.857	0.3672245\\
0.858	0.368082\\
0.859	0.3689405\\
0.86	0.3698\\
0.861	0.3706605\\
0.862	0.371522\\
0.863	0.3723845\\
0.864	0.373248\\
0.865	0.3741125\\
0.866	0.374978\\
0.867	0.3758445\\
0.868	0.376712\\
0.869	0.3775805\\
0.87	0.37845\\
0.871	0.3793205\\
0.872	0.380192\\
0.873	0.3810645\\
0.874	0.381938\\
0.875	0.3828125\\
0.876	0.383688\\
0.877	0.3845645\\
0.878	0.385442\\
0.879	0.3863205\\
0.88	0.3872\\
0.881	0.3880805\\
0.882	0.388962\\
0.883	0.3898445\\
0.884	0.390728\\
0.885	0.3916125\\
0.886	0.392498\\
0.887	0.3933845\\
0.888	0.394272\\
0.889	0.3951605\\
0.89	0.39605\\
0.891	0.3969405\\
0.892	0.397832\\
0.893	0.3987245\\
0.894	0.399618\\
0.895	0.4005125\\
0.896	0.401408\\
0.897	0.4023045\\
0.898	0.403202\\
0.899	0.4041005\\
0.9	0.405\\
0.901	0.4059005\\
0.902	0.406802\\
0.903	0.4077045\\
0.904	0.408608\\
0.905	0.4095125\\
0.906	0.410418\\
0.907	0.4113245\\
0.908	0.412232\\
0.909	0.4131405\\
0.91	0.41405\\
0.911	0.4149605\\
0.912	0.415872\\
0.913	0.4167845\\
0.914	0.417698\\
0.915	0.4186125\\
0.916	0.419528\\
0.917	0.4204445\\
0.918	0.421362\\
0.919	0.4222805\\
0.92	0.4232\\
0.921	0.4241205\\
0.922	0.425042\\
0.923	0.4259645\\
0.924	0.426888\\
0.925	0.4278125\\
0.926	0.428738\\
0.927	0.4296645\\
0.928	0.430592\\
0.929	0.4315205\\
0.93	0.43245\\
0.931	0.4333805\\
0.932	0.434312\\
0.933	0.4352445\\
0.934	0.436178\\
0.935	0.4371125\\
0.936	0.438048\\
0.937	0.4389845\\
0.938	0.439922\\
0.939	0.4408605\\
0.94	0.4418\\
0.941	0.4427405\\
0.942	0.443682\\
0.943	0.4446245\\
0.944	0.445568\\
0.945	0.4465125\\
0.946	0.447458\\
0.947	0.4484045\\
0.948	0.449352\\
0.949	0.4503005\\
0.95	0.45125\\
0.951	0.452567373019925\\
0.952	0.454336036183179\\
0.953	0.456537114536876\\
0.954	0.45915173312813\\
0.955	0.462161017004055\\
0.956	0.465546091211763\\
0.957	0.469288080798368\\
0.958	0.473368110810984\\
0.959	0.477767306296726\\
0.96	0.482466792302704\\
0.961	0.487447693876033\\
0.962	0.492691136063827\\
0.963	0.498178243913202\\
0.964	0.503890142471266\\
0.965	0.509807956785136\\
0.966	0.515912811901924\\
0.967	0.522185832868748\\
0.968	0.528608144732715\\
0.969	0.535160872540941\\
0.97	0.541825141340541\\
0.971	0.548582076178629\\
0.972	0.555412802102315\\
0.973	0.562298444158714\\
0.974	0.56922012739494\\
0.975	0.57615897685811\\
0.976	0.583096117595331\\
0.977	0.59001267465372\\
0.978	0.596889773080389\\
0.979	0.603708537922456\\
0.98	0.610450094227028\\
0.981	0.617095567041222\\
0.982	0.623626081412151\\
0.983	0.630022762386931\\
0.984	0.636266735012671\\
0.985	0.642339124336487\\
0.986	0.648221055405491\\
0.987	0.653893653266801\\
0.988	0.659338042967525\\
0.989	0.664535349554778\\
0.99	0.669466698075675\\
0.991	0.674113213577331\\
0.992	0.678456021106854\\
0.993	0.682476245711362\\
0.994	0.686155012437967\\
0.995	0.689473446333784\\
0.996	0.692412672445925\\
0.997	0.694953815821503\\
0.998	0.697078001507632\\
0.999	0.698766354551427\\
1	0.7\\
1.001	0.700766354551427\\
1.002	0.701078001507632\\
1.003	0.700953815821503\\
1.004	0.700412672445924\\
1.005	0.699473446333784\\
1.006	0.698155012437968\\
1.007	0.696476245711362\\
1.008	0.694456021106854\\
1.009	0.69211321357733\\
1.01	0.689466698075676\\
1.011	0.686535349554778\\
1.012	0.683338042967524\\
1.013	0.6798936532668\\
1.014	0.676221055405492\\
1.015	0.672339124336486\\
1.016	0.66826673501267\\
1.017	0.66402276238693\\
1.018	0.659626081412152\\
1.019	0.655095567041221\\
1.02	0.650450094227027\\
1.021	0.645708537922454\\
1.022	0.64088977308039\\
1.023	0.636012674653718\\
1.024	0.631096117595329\\
1.025	0.626158976858108\\
1.026	0.621220127394941\\
1.027	0.616298444158713\\
1.028	0.611412802102313\\
1.029	0.606582076178627\\
1.03	0.601825141340541\\
1.031	0.59716087254094\\
1.032	0.592608144732713\\
1.033	0.588185832868746\\
1.034	0.583912811901925\\
1.035	0.579807956785134\\
1.036	0.575890142471265\\
1.037	0.5721782439132\\
1.038	0.568691136063828\\
1.039	0.565447693876032\\
1.04	0.562466792302703\\
1.041	0.559767306296725\\
1.042	0.557368110810984\\
1.043	0.555288080798367\\
1.044	0.553546091211762\\
1.045	0.552161017004054\\
1.046	0.55115173312813\\
1.047	0.550537114536876\\
1.048	0.550336036183178\\
1.049	0.550567373019924\\
1.05	0.55125\\
1.051	0.5523005\\
1.052	0.553352\\
1.053	0.5544045\\
1.054	0.555458\\
1.055	0.5565125\\
1.056	0.557568\\
1.057	0.5586245\\
1.058	0.559682\\
1.059	0.5607405\\
1.06	0.5618\\
1.061	0.5628605\\
1.062	0.563922\\
1.063	0.5649845\\
1.064	0.566048\\
1.065	0.5671125\\
1.066	0.568178\\
1.067	0.5692445\\
1.068	0.570312\\
1.069	0.5713805\\
1.07	0.57245\\
1.071	0.5735205\\
1.072	0.574592\\
1.073	0.5756645\\
1.074	0.576738\\
1.075	0.5778125\\
1.076	0.578888\\
1.077	0.5799645\\
1.078	0.581042\\
1.079	0.5821205\\
1.08	0.5832\\
1.081	0.5842805\\
1.082	0.585362\\
1.083	0.5864445\\
1.084	0.587528\\
1.085	0.5886125\\
1.086	0.589698\\
1.087	0.5907845\\
1.088	0.591872\\
1.089	0.5929605\\
1.09	0.59405\\
1.091	0.5951405\\
1.092	0.596232\\
1.093	0.5973245\\
1.094	0.598418\\
1.095	0.5995125\\
1.096	0.600608\\
1.097	0.6017045\\
1.098	0.602802\\
1.099	0.6039005\\
1.1	0.605\\
1.101	0.6061005\\
1.102	0.607202\\
1.103	0.6083045\\
1.104	0.609408\\
1.105	0.6105125\\
1.106	0.611618\\
1.107	0.6127245\\
1.108	0.613832\\
1.109	0.6149405\\
1.11	0.61605\\
1.111	0.6171605\\
1.112	0.618272\\
1.113	0.6193845\\
1.114	0.620498\\
1.115	0.6216125\\
1.116	0.622728\\
1.117	0.6238445\\
1.118	0.624962\\
1.119	0.6260805\\
1.12	0.6272\\
1.121	0.6283205\\
1.122	0.629442\\
1.123	0.6305645\\
1.124	0.631688\\
1.125	0.6328125\\
1.126	0.633938\\
1.127	0.6350645\\
1.128	0.636192\\
1.129	0.6373205\\
1.13	0.63845\\
1.131	0.6395805\\
1.132	0.640712\\
1.133	0.6418445\\
1.134	0.642978\\
1.135	0.6441125\\
1.136	0.645248\\
1.137	0.6463845\\
1.138	0.647522\\
1.139	0.6486605\\
1.14	0.6498\\
1.141	0.6509405\\
1.142	0.652082\\
1.143	0.6532245\\
1.144	0.654368\\
1.145	0.6555125\\
1.146	0.656658\\
1.147	0.6578045\\
1.148	0.658952\\
1.149	0.6601005\\
1.15	0.66125\\
1.151	0.6624005\\
1.152	0.663552\\
1.153	0.6647045\\
1.154	0.665858\\
1.155	0.6670125\\
1.156	0.668168\\
1.157	0.6693245\\
1.158	0.670482\\
1.159	0.6716405\\
1.16	0.6728\\
1.161	0.6739605\\
1.162	0.675122\\
1.163	0.6762845\\
1.164	0.677448\\
1.165	0.6786125\\
1.166	0.679778\\
1.167	0.6809445\\
1.168	0.682112\\
1.169	0.6832805\\
1.17	0.68445\\
1.171	0.6856205\\
1.172	0.686792\\
1.173	0.6879645\\
1.174	0.689138\\
1.175	0.6903125\\
1.176	0.691488\\
1.177	0.6926645\\
1.178	0.693842\\
1.179	0.6950205\\
1.18	0.6962\\
1.181	0.6973805\\
1.182	0.698562\\
1.183	0.6997445\\
1.184	0.700928\\
1.185	0.7021125\\
1.186	0.703298\\
1.187	0.7044845\\
1.188	0.705672\\
1.189	0.7068605\\
1.19	0.70805\\
1.191	0.7092405\\
1.192	0.710432\\
1.193	0.7116245\\
1.194	0.712818\\
1.195	0.7140125\\
1.196	0.715208\\
1.197	0.7164045\\
1.198	0.717602\\
1.199	0.7188005\\
1.2	0.72\\
1.201	0.7212005\\
1.202	0.722402\\
1.203	0.7236045\\
1.204	0.724808\\
1.205	0.7260125\\
1.206	0.727218\\
1.207	0.7284245\\
1.208	0.729632\\
1.209	0.7308405\\
1.21	0.73205\\
1.211	0.7332605\\
1.212	0.734472\\
1.213	0.7356845\\
1.214	0.736898\\
1.215	0.7381125\\
1.216	0.739328\\
1.217	0.7405445\\
1.218	0.741762\\
1.219	0.7429805\\
1.22	0.7442\\
1.221	0.7454205\\
1.222	0.746642\\
1.223	0.7478645\\
1.224	0.749088\\
1.225	0.7503125\\
1.226	0.751538\\
1.227	0.7527645\\
1.228	0.753992\\
1.229	0.7552205\\
1.23	0.75645\\
1.231	0.7576805\\
1.232	0.758912\\
1.233	0.7601445\\
1.234	0.761378\\
1.235	0.7626125\\
1.236	0.763848\\
1.237	0.7650845\\
1.238	0.766322\\
1.239	0.7675605\\
1.24	0.7688\\
1.241	0.7700405\\
1.242	0.771282\\
1.243	0.7725245\\
1.244	0.773768\\
1.245	0.7750125\\
1.246	0.776258\\
1.247	0.7775045\\
1.248	0.778752\\
1.249	0.7800005\\
1.25	0.78125\\
1.251	0.7825005\\
1.252	0.783752\\
1.253	0.7850045\\
1.254	0.786258\\
1.255	0.7875125\\
1.256	0.788768\\
1.257	0.7900245\\
1.258	0.791282\\
1.259	0.7925405\\
1.26	0.7938\\
1.261	0.7950605\\
1.262	0.796322\\
1.263	0.7975845\\
1.264	0.798848\\
1.265	0.8001125\\
1.266	0.801378\\
1.267	0.8026445\\
1.268	0.803912\\
1.269	0.8051805\\
1.27	0.80645\\
1.271	0.8077205\\
1.272	0.808992\\
1.273	0.8102645\\
1.274	0.811538\\
1.275	0.8128125\\
1.276	0.814088\\
1.277	0.8153645\\
1.278	0.816642\\
1.279	0.8179205\\
1.28	0.8192\\
1.281	0.8204805\\
1.282	0.821762\\
1.283	0.8230445\\
1.284	0.824328\\
1.285	0.8256125\\
1.286	0.826898\\
1.287	0.8281845\\
1.288	0.829472\\
1.289	0.8307605\\
1.29	0.83205\\
1.291	0.8333405\\
1.292	0.834632\\
1.293	0.8359245\\
1.294	0.837218\\
1.295	0.8385125\\
1.296	0.839808\\
1.297	0.8411045\\
1.298	0.842402\\
1.299	0.8437005\\
1.3	0.845\\
1.301	0.8463005\\
1.302	0.847602\\
1.303	0.8489045\\
1.304	0.850208\\
1.305	0.8515125\\
1.306	0.852818\\
1.307	0.8541245\\
1.308	0.855432\\
1.309	0.8567405\\
1.31	0.85805\\
1.311	0.8593605\\
1.312	0.860672\\
1.313	0.8619845\\
1.314	0.863298\\
1.315	0.8646125\\
1.316	0.865928\\
1.317	0.8672445\\
1.318	0.868562\\
1.319	0.8698805\\
1.32	0.8712\\
1.321	0.8725205\\
1.322	0.873842\\
1.323	0.8751645\\
1.324	0.876488\\
1.325	0.8778125\\
1.326	0.879138\\
1.327	0.8804645\\
1.328	0.881792\\
1.329	0.8831205\\
1.33	0.88445\\
1.331	0.8857805\\
1.332	0.887112\\
1.333	0.8884445\\
1.334	0.889778\\
1.335	0.8911125\\
1.336	0.892448\\
1.337	0.8937845\\
1.338	0.895122\\
1.339	0.8964605\\
1.34	0.8978\\
1.341	0.8991405\\
1.342	0.900482\\
1.343	0.9018245\\
1.344	0.903168\\
1.345	0.9045125\\
1.346	0.905858\\
1.347	0.9072045\\
1.348	0.908552\\
1.349	0.9099005\\
1.35	0.91125\\
1.351	0.9126005\\
1.352	0.913952\\
1.353	0.9153045\\
1.354	0.916658\\
1.355	0.9180125\\
1.356	0.919368\\
1.357	0.9207245\\
1.358	0.922082\\
1.359	0.9234405\\
1.36	0.9248\\
1.361	0.9261605\\
1.362	0.927522\\
1.363	0.9288845\\
1.364	0.930248\\
1.365	0.9316125\\
1.366	0.932978\\
1.367	0.9343445\\
1.368	0.935712\\
1.369	0.9370805\\
1.37	0.93845\\
1.371	0.9398205\\
1.372	0.941192\\
1.373	0.9425645\\
1.374	0.943938\\
1.375	0.9453125\\
1.376	0.946688\\
1.377	0.9480645\\
1.378	0.949442\\
1.379	0.9508205\\
1.38	0.9522\\
1.381	0.9535805\\
1.382	0.954962\\
1.383	0.9563445\\
1.384	0.957728\\
1.385	0.9591125\\
1.386	0.960498\\
1.387	0.9618845\\
1.388	0.963272\\
1.389	0.9646605\\
1.39	0.96605\\
1.391	0.9674405\\
1.392	0.968832\\
1.393	0.9702245\\
1.394	0.971618\\
1.395	0.9730125\\
1.396	0.974408\\
1.397	0.9758045\\
1.398	0.977202\\
1.399	0.9786005\\
1.4	0.98\\
1.401	0.9814005\\
1.402	0.982802\\
1.403	0.9842045\\
1.404	0.985608\\
1.405	0.9870125\\
1.406	0.988418\\
1.407	0.9898245\\
1.408	0.991232\\
1.409	0.9926405\\
1.41	0.99405\\
1.411	0.9954605\\
1.412	0.996872\\
1.413	0.9982845\\
1.414	0.999698\\
1.415	1.0011125\\
1.416	1.002528\\
1.417	1.0039445\\
1.418	1.005362\\
1.419	1.0067805\\
1.42	1.0082\\
1.421	1.0096205\\
1.422	1.011042\\
1.423	1.0124645\\
1.424	1.013888\\
1.425	1.0153125\\
1.426	1.016738\\
1.427	1.0181645\\
1.428	1.019592\\
1.429	1.0210205\\
1.43	1.02245\\
1.431	1.0238805\\
1.432	1.025312\\
1.433	1.0267445\\
1.434	1.028178\\
1.435	1.0296125\\
1.436	1.031048\\
1.437	1.0324845\\
1.438	1.033922\\
1.439	1.0353605\\
1.44	1.0368\\
1.441	1.0382405\\
1.442	1.039682\\
1.443	1.0411245\\
1.444	1.042568\\
1.445	1.0440125\\
1.446	1.045458\\
1.447	1.0469045\\
1.448	1.048352\\
1.449	1.0498005\\
1.45	1.05125\\
1.451	1.0527005\\
1.452	1.054152\\
1.453	1.0556045\\
1.454	1.057058\\
1.455	1.0585125\\
1.456	1.059968\\
1.457	1.0614245\\
1.458	1.062882\\
1.459	1.0643405\\
1.46	1.0658\\
1.461	1.0672605\\
1.462	1.068722\\
1.463	1.0701845\\
1.464	1.071648\\
1.465	1.0731125\\
1.466	1.074578\\
1.467	1.0760445\\
1.468	1.077512\\
1.469	1.0789805\\
1.47	1.08045\\
1.471	1.0819205\\
1.472	1.083392\\
1.473	1.0848645\\
1.474	1.086338\\
1.475	1.0878125\\
1.476	1.089288\\
1.477	1.0907645\\
1.478	1.092242\\
1.479	1.0937205\\
1.48	1.0952\\
1.481	1.0966805\\
1.482	1.098162\\
1.483	1.0996445\\
1.484	1.101128\\
1.485	1.1026125\\
1.486	1.104098\\
1.487	1.1055845\\
1.488	1.107072\\
1.489	1.1085605\\
1.49	1.11005\\
1.491	1.1115405\\
1.492	1.113032\\
1.493	1.1145245\\
1.494	1.116018\\
1.495	1.1175125\\
1.496	1.119008\\
1.497	1.1205045\\
1.498	1.122002\\
1.499	1.1235005\\
1.5	1.125\\
1.501	1.1265005\\
1.502	1.128002\\
1.503	1.1295045\\
1.504	1.131008\\
1.505	1.1325125\\
1.506	1.134018\\
1.507	1.1355245\\
1.508	1.137032\\
1.509	1.1385405\\
1.51	1.14005\\
1.511	1.1415605\\
1.512	1.143072\\
1.513	1.1445845\\
1.514	1.146098\\
1.515	1.1476125\\
1.516	1.149128\\
1.517	1.1506445\\
1.518	1.152162\\
1.519	1.1536805\\
1.52	1.1552\\
1.521	1.1567205\\
1.522	1.158242\\
1.523	1.1597645\\
1.524	1.161288\\
1.525	1.1628125\\
1.526	1.164338\\
1.527	1.1658645\\
1.528	1.167392\\
1.529	1.1689205\\
1.53	1.17045\\
1.531	1.1719805\\
1.532	1.173512\\
1.533	1.1750445\\
1.534	1.176578\\
1.535	1.1781125\\
1.536	1.179648\\
1.537	1.1811845\\
1.538	1.182722\\
1.539	1.1842605\\
1.54	1.1858\\
1.541	1.1873405\\
1.542	1.188882\\
1.543	1.1904245\\
1.544	1.191968\\
1.545	1.1935125\\
1.546	1.195058\\
1.547	1.1966045\\
1.548	1.198152\\
1.549	1.1997005\\
1.55	1.20125\\
1.551	1.2028005\\
1.552	1.204352\\
1.553	1.2059045\\
1.554	1.207458\\
1.555	1.2090125\\
1.556	1.210568\\
1.557	1.2121245\\
1.558	1.213682\\
1.559	1.2152405\\
1.56	1.2168\\
1.561	1.2183605\\
1.562	1.219922\\
1.563	1.2214845\\
1.564	1.223048\\
1.565	1.2246125\\
1.566	1.226178\\
1.567	1.2277445\\
1.568	1.229312\\
1.569	1.2308805\\
1.57	1.23245\\
1.571	1.2340205\\
1.572	1.235592\\
1.573	1.2371645\\
1.574	1.238738\\
1.575	1.2403125\\
1.576	1.241888\\
1.577	1.2434645\\
1.578	1.245042\\
1.579	1.2466205\\
1.58	1.2482\\
1.581	1.2497805\\
1.582	1.251362\\
1.583	1.2529445\\
1.584	1.254528\\
1.585	1.2561125\\
1.586	1.257698\\
1.587	1.2592845\\
1.588	1.260872\\
1.589	1.2624605\\
1.59	1.26405\\
1.591	1.2656405\\
1.592	1.267232\\
1.593	1.2688245\\
1.594	1.270418\\
1.595	1.2720125\\
1.596	1.273608\\
1.597	1.2752045\\
1.598	1.276802\\
1.599	1.2784005\\
1.6	1.28\\
1.601	1.2816005\\
1.602	1.283202\\
1.603	1.2848045\\
1.604	1.286408\\
1.605	1.2880125\\
1.606	1.289618\\
1.607	1.2912245\\
1.608	1.292832\\
1.609	1.2944405\\
1.61	1.29605\\
1.611	1.2976605\\
1.612	1.299272\\
1.613	1.3008845\\
1.614	1.302498\\
1.615	1.3041125\\
1.616	1.305728\\
1.617	1.3073445\\
1.618	1.308962\\
1.619	1.3105805\\
1.62	1.3122\\
1.621	1.3138205\\
1.622	1.315442\\
1.623	1.3170645\\
1.624	1.318688\\
1.625	1.3203125\\
1.626	1.321938\\
1.627	1.3235645\\
1.628	1.325192\\
1.629	1.3268205\\
1.63	1.32845\\
1.631	1.3300805\\
1.632	1.331712\\
1.633	1.3333445\\
1.634	1.334978\\
1.635	1.3366125\\
1.636	1.338248\\
1.637	1.3398845\\
1.638	1.341522\\
1.639	1.3431605\\
1.64	1.3448\\
1.641	1.3464405\\
1.642	1.348082\\
1.643	1.3497245\\
1.644	1.351368\\
1.645	1.3530125\\
1.646	1.354658\\
1.647	1.3563045\\
1.648	1.357952\\
1.649	1.3596005\\
1.65	1.36125\\
1.651	1.3629005\\
1.652	1.364552\\
1.653	1.3662045\\
1.654	1.367858\\
1.655	1.3695125\\
1.656	1.371168\\
1.657	1.3728245\\
1.658	1.374482\\
1.659	1.3761405\\
1.66	1.3778\\
1.661	1.3794605\\
1.662	1.381122\\
1.663	1.3827845\\
1.664	1.384448\\
1.665	1.3861125\\
1.666	1.387778\\
1.667	1.3894445\\
1.668	1.391112\\
1.669	1.3927805\\
1.67	1.39445\\
1.671	1.3961205\\
1.672	1.397792\\
1.673	1.3994645\\
1.674	1.401138\\
1.675	1.4028125\\
1.676	1.404488\\
1.677	1.4061645\\
1.678	1.407842\\
1.679	1.4095205\\
1.68	1.4112\\
1.681	1.4128805\\
1.682	1.414562\\
1.683	1.4162445\\
1.684	1.417928\\
1.685	1.4196125\\
1.686	1.421298\\
1.687	1.4229845\\
1.688	1.424672\\
1.689	1.4263605\\
1.69	1.42805\\
1.691	1.4297405\\
1.692	1.431432\\
1.693	1.4331245\\
1.694	1.434818\\
1.695	1.4365125\\
1.696	1.438208\\
1.697	1.4399045\\
1.698	1.441602\\
1.699	1.4433005\\
1.7	1.445\\
1.701	1.4467005\\
1.702	1.448402\\
1.703	1.4501045\\
1.704	1.451808\\
1.705	1.4535125\\
1.706	1.455218\\
1.707	1.4569245\\
1.708	1.458632\\
1.709	1.4603405\\
1.71	1.46205\\
1.711	1.4637605\\
1.712	1.465472\\
1.713	1.4671845\\
1.714	1.468898\\
1.715	1.4706125\\
1.716	1.472328\\
1.717	1.4740445\\
1.718	1.475762\\
1.719	1.4774805\\
1.72	1.4792\\
1.721	1.4809205\\
1.722	1.482642\\
1.723	1.4843645\\
1.724	1.486088\\
1.725	1.4878125\\
1.726	1.489538\\
1.727	1.4912645\\
1.728	1.492992\\
1.729	1.4947205\\
1.73	1.49645\\
1.731	1.4981805\\
1.732	1.499912\\
1.733	1.5016445\\
1.734	1.503378\\
1.735	1.5051125\\
1.736	1.506848\\
1.737	1.5085845\\
1.738	1.510322\\
1.739	1.5120605\\
1.74	1.5138\\
1.741	1.5155405\\
1.742	1.517282\\
1.743	1.5190245\\
1.744	1.520768\\
1.745	1.5225125\\
1.746	1.524258\\
1.747	1.5260045\\
1.748	1.527752\\
1.749	1.5295005\\
1.75	1.53125\\
1.751	1.5330005\\
1.752	1.534752\\
1.753	1.5365045\\
1.754	1.538258\\
1.755	1.5400125\\
1.756	1.541768\\
1.757	1.5435245\\
1.758	1.545282\\
1.759	1.5470405\\
1.76	1.5488\\
1.761	1.5505605\\
1.762	1.552322\\
1.763	1.5540845\\
1.764	1.555848\\
1.765	1.5576125\\
1.766	1.559378\\
1.767	1.5611445\\
1.768	1.562912\\
1.769	1.5646805\\
1.77	1.56645\\
1.771	1.5682205\\
1.772	1.569992\\
1.773	1.5717645\\
1.774	1.573538\\
1.775	1.5753125\\
1.776	1.577088\\
1.777	1.5788645\\
1.778	1.580642\\
1.779	1.5824205\\
1.78	1.5842\\
1.781	1.5859805\\
1.782	1.587762\\
1.783	1.5895445\\
1.784	1.591328\\
1.785	1.5931125\\
1.786	1.594898\\
1.787	1.5966845\\
1.788	1.598472\\
1.789	1.6002605\\
1.79	1.60205\\
1.791	1.6038405\\
1.792	1.605632\\
1.793	1.6074245\\
1.794	1.609218\\
1.795	1.6110125\\
1.796	1.612808\\
1.797	1.6146045\\
1.798	1.616402\\
1.799	1.6182005\\
1.8	1.62\\
1.801	1.6218005\\
1.802	1.623602\\
1.803	1.6254045\\
1.804	1.627208\\
1.805	1.6290125\\
1.806	1.630818\\
1.807	1.6326245\\
1.808	1.634432\\
1.809	1.6362405\\
1.81	1.63805\\
1.811	1.6398605\\
1.812	1.641672\\
1.813	1.6434845\\
1.814	1.645298\\
1.815	1.6471125\\
1.816	1.648928\\
1.817	1.6507445\\
1.818	1.652562\\
1.819	1.6543805\\
1.82	1.6562\\
1.821	1.6580205\\
1.822	1.659842\\
1.823	1.6616645\\
1.824	1.663488\\
1.825	1.6653125\\
1.826	1.667138\\
1.827	1.6689645\\
1.828	1.670792\\
1.829	1.6726205\\
1.83	1.67445\\
1.831	1.6762805\\
1.832	1.678112\\
1.833	1.6799445\\
1.834	1.681778\\
1.835	1.6836125\\
1.836	1.685448\\
1.837	1.6872845\\
1.838	1.689122\\
1.839	1.6909605\\
1.84	1.6928\\
1.841	1.6946405\\
1.842	1.696482\\
1.843	1.6983245\\
1.844	1.700168\\
1.845	1.7020125\\
1.846	1.703858\\
1.847	1.7057045\\
1.848	1.707552\\
1.849	1.7094005\\
1.85	1.71125\\
1.851	1.7131005\\
1.852	1.714952\\
1.853	1.7168045\\
1.854	1.718658\\
1.855	1.7205125\\
1.856	1.722368\\
1.857	1.7242245\\
1.858	1.726082\\
1.859	1.7279405\\
1.86	1.7298\\
1.861	1.7316605\\
1.862	1.733522\\
1.863	1.7353845\\
1.864	1.737248\\
1.865	1.7391125\\
1.866	1.740978\\
1.867	1.7428445\\
1.868	1.744712\\
1.869	1.7465805\\
1.87	1.74845\\
1.871	1.7503205\\
1.872	1.752192\\
1.873	1.7540645\\
1.874	1.755938\\
1.875	1.7578125\\
1.876	1.759688\\
1.877	1.7615645\\
1.878	1.763442\\
1.879	1.7653205\\
1.88	1.7672\\
1.881	1.7690805\\
1.882	1.770962\\
1.883	1.7728445\\
1.884	1.774728\\
1.885	1.7766125\\
1.886	1.778498\\
1.887	1.7803845\\
1.888	1.782272\\
1.889	1.7841605\\
1.89	1.78605\\
1.891	1.7879405\\
1.892	1.789832\\
1.893	1.7917245\\
1.894	1.793618\\
1.895	1.7955125\\
1.896	1.797408\\
1.897	1.7993045\\
1.898	1.801202\\
1.899	1.8031005\\
1.9	1.805\\
1.901	1.8069005\\
1.902	1.808802\\
1.903	1.8107045\\
1.904	1.812608\\
1.905	1.8145125\\
1.906	1.816418\\
1.907	1.8183245\\
1.908	1.820232\\
1.909	1.8221405\\
1.91	1.82405\\
1.911	1.8259605\\
1.912	1.827872\\
1.913	1.8297845\\
1.914	1.831698\\
1.915	1.8336125\\
1.916	1.835528\\
1.917	1.8374445\\
1.918	1.839362\\
1.919	1.8412805\\
1.92	1.8432\\
1.921	1.8451205\\
1.922	1.847042\\
1.923	1.8489645\\
1.924	1.850888\\
1.925	1.8528125\\
1.926	1.854738\\
1.927	1.8566645\\
1.928	1.858592\\
1.929	1.8605205\\
1.93	1.86245\\
1.931	1.8643805\\
1.932	1.866312\\
1.933	1.8682445\\
1.934	1.870178\\
1.935	1.8721125\\
1.936	1.874048\\
1.937	1.8759845\\
1.938	1.877922\\
1.939	1.8798605\\
1.94	1.8818\\
1.941	1.8837405\\
1.942	1.885682\\
1.943	1.8876245\\
1.944	1.889568\\
1.945	1.8915125\\
1.946	1.893458\\
1.947	1.8954045\\
1.948	1.897352\\
1.949	1.8993005\\
1.95	1.90125\\
1.951	1.9032005\\
1.952	1.905152\\
1.953	1.9071045\\
1.954	1.909058\\
1.955	1.9110125\\
1.956	1.912968\\
1.957	1.9149245\\
1.958	1.916882\\
1.959	1.9188405\\
1.96	1.9208\\
1.961	1.9227605\\
1.962	1.924722\\
1.963	1.9266845\\
1.964	1.928648\\
1.965	1.9306125\\
1.966	1.932578\\
1.967	1.9345445\\
1.968	1.936512\\
1.969	1.9384805\\
1.97	1.94045\\
1.971	1.9424205\\
1.972	1.944392\\
1.973	1.9463645\\
1.974	1.948338\\
1.975	1.9503125\\
1.976	1.952288\\
1.977	1.9542645\\
1.978	1.956242\\
1.979	1.9582205\\
1.98	1.9602\\
1.981	1.9621805\\
1.982	1.964162\\
1.983	1.9661445\\
1.984	1.968128\\
1.985	1.9701125\\
1.986	1.972098\\
1.987	1.9740845\\
1.988	1.976072\\
1.989	1.9780605\\
1.99	1.98005\\
1.991	1.9820405\\
1.992	1.984032\\
1.993	1.9860245\\
1.994	1.988018\\
1.995	1.9900125\\
1.996	1.992008\\
1.997	1.9940045\\
1.998	1.996002\\
1.999	1.9980005\\
2	2\\
};
\addplot [color=black, forget plot]
  table[row sep=crcr]{%
4	8\\
};
\end{axis}
\end{tikzpicture}%